\documentclass[aps,prx,superscriptaddress,twocolumn,nofootinbib,10pt,longbibliography]{revtex4-2}
\usepackage[utf8]{inputenc}

\makeatletter
\renewcommand{\bibinfo}[2]{%
  \ifstrequal{#1}{title}{\textit{#2}}{#2}}
\makeatother

\usepackage{mathpazo}

\usepackage[titletoc,toc,title,page]{appendix}

\usepackage{csquotes}
\usepackage[italian,english]{babel}

\usepackage{microtype}

\usepackage{bm}
\usepackage{dsfont} 
\usepackage{amsmath,amssymb,amsthm,thmtools}
\usepackage{mathtools}
\usepackage{cases}
\usepackage{calc}
\usepackage{mathrsfs} 
\usepackage[normalem]{ulem} 

\usepackage{minitoc}

\usepackage{parskip}
\usepackage[colorlinks=true]{hyperref}
\usepackage[nameinlink]{cleveref}
\crefname{conjecture}{conjecture}{conjectures}
\crefname{box}{Box}{Box}
\hypersetup{
  colorlinks   = true, 
  urlcolor     = green!80!black, 
  linkcolor    = blue, 
  citecolor    = red!80!black 
}

\usepackage{orcidlink}

\usepackage{physics} 

\usepackage{graphicx}

\usepackage[dvipsnames,table]{xcolor}

\usepackage{placeins}
\usepackage{multirow,tabularx,booktabs}
\usepackage{tcolorbox}
\tcbuselibrary{theorems,breakable}

\graphicspath{{./figures/}}

\newcommand{\calH}{\mathcal{H}}

\newcommand{\calP}{\mathcal{P}}
\newcommand{\calQ}{\mathcal{Q}}
\newcommand{\calS}{\mathcal{S}}

\newcommand{\calZ}{\mathcal{Z}}

\newcommand{\ps}[1]{\mathtt{#1}}

\newcommand{\tildecalP}{\widetilde{\calP}}
\newcommand{\tildecalQ}{\widetilde{\calQ}}
\newcommand{\tildecalS}{\widetilde{\calS}}

\newcommand{\nout}{n_{\rm out}}

\newcommand{\bs}[1]{\boldsymbol{#1}}
\newcommand{\on}[1]{\operatorname{#1}}
\newcommand{\dimspanmu}{s_{\bs\mu}}

\newcommand{\parTitle}[1]{\textit{#1 ---}}

\newtheorem{theorem}{Theorem}

\newtheorem{lemma}{Lemma}
\newtheorem{corollary}{Corollary}
\newtheorem{conjecture}{Conjecture}
\newtheorem{result}{Result}
\theoremstyle{definition}
\newtcbtheorem[crefname={example}{examples}]{example}{Example}%
{
  parbox=false,
  breakable,
  colback=red!0,          
  colframe=gray!10,         
  fonttitle=\bfseries,      
  arc=2mm,                  
  boxrule=0.2pt,           
  left=3mm,                
  right=3mm,               
  top=2mm,                 
  bottom=2mm,
  fonttitle=\bfseries\color{black}
}{}

\begin{document}

\title{The non-stabilizerness cost of quantum state estimation}

\author{Gabriele Lo Monaco\,\orcidlink{0000-0002-3594-3477}}
\email{gabriele.lomonaco@unipa.it}
\let\comma,
\affiliation{Universit\`a degli Studi di Palermo\comma{} Dipartimento di Fisica e Chimica - Emilio Segr\`e\comma{} via Archirafi 36\comma{} I-90123 Palermo\comma{} Italy}

\author{Salvatore Lorenzo\,\orcidlink{0000-0002-0827-5549}}
\let\comma,
\affiliation{Universit\`a degli Studi di Palermo\comma{} Dipartimento di Fisica e Chimica - Emilio Segr\`e\comma{} via Archirafi 36\comma{} I-90123 Palermo\comma{} Italy}

\author{Alessandro Ferraro\,\orcidlink{0000-0002-7579-6336}}
\affiliation{Quantum Technology Lab\comma{} Dipartimento di Fisica Aldo Pontremoli\comma{} Universit\`a degli Studi di Milano\comma{} I-20133 Milano\comma{} Italy}

\author{Mauro Paternostro\,\orcidlink{0000-0001-8870-9134}}
\affiliation{Universit\`a degli Studi di Palermo\comma{} Dipartimento di Fisica e Chimica - Emilio Segr\`e\comma{} via Archirafi 36\comma{} I-90123 Palermo\comma{} Italy}
\affiliation{Centre for Quantum Materials and Technologies\comma{} School of Mathematics and Physics\comma{} Queen's University Belfast\comma{} BT7 1NN\comma{} United Kingdom}

\author{G. Massimo Palma\,\orcidlink{0000-0001-7009-4573}}
\let\comma,
\affiliation{Universit\`a degli Studi di Palermo\comma{} Dipartimento di Fisica e Chimica - Emilio Segr\`e\comma{} via Archirafi 36\comma{} I-90123 Palermo\comma{} Italy}

\author{Luca Innocenti\,\orcidlink{0000-0002-7678-1128}}
\let\comma,
\affiliation{Universit\`a degli Studi di Palermo\comma{} Dipartimento di Fisica e Chimica - Emilio Segr\`e\comma{} via Archirafi 36\comma{} I-90123 Palermo\comma{} Italy}

\begin{abstract}
\noindent
We study the non‑stabilizer resources required to achieve informational completeness in single-setting quantum state estimation scenarios.
We consider fixed-basis projective measurements preceded by quantum circuits acting on $n$-qubit input states, allowing ancillary qubits to increase retrievable information.
We prove that when only stabilizer resources are allowed, these strategies are always informationally equivalent to projective measurements in a stabilizer basis, and therefore never informationally complete, regardless of the number of ancillas.
We then show that incorporating $T$ gates enlarges the accessible information.
Specifically, we prove that at least ${2n}/{\log_2 3}$ such gates are necessary for informational completeness, and that $2n$ suffice.
We conjecture that $2n$ gates are indeed both necessary and sufficient.
Finally, we unveil a tight connection between entanglement structure and informational power of measurements implemented with $t$-doped Clifford circuits.
Our results recast notions of ``magic'' and stabilizerness --- typically framed in computational terms --- into the setting of quantum metrology.
\end{abstract}

\maketitle

Reconstruction protocols such as quantum reservoir computing~\cite{tanaka2019recent,fujii2020quantum,mujal2021opportunities,chen2020temporal}, quantum extreme learning machines (QELMs)~\cite{fujii2017harnessing,%
ghosh2019quantum,ghosh2021realising,
kutvonen2020optimizing,tran2020higherorder,krisnanda2021creating,%
rafayelyan2020largescale,rafayelyan2020largescale,
nokkala2021gaussian,
nakajima2019boosting,innocenti2023potential,lomonaco_extremePES,vetrano2025state,suprano2024experimental,zia2502quantum,xiong2023fundamental}, shadow tomography~\cite{aaronson2018ShadowTomographyQuantum,huang2020predicting,acharya2021ShadowTomographyBased,nguyen2022OptimizingShadowTomography,elben2022RandomizedMeasurementToolbox,innocenti2023shadow}, and even conventional linear state tomography~\cite{paris2004QuantumStateEstimation,dariano2003QuantumTomography,teo2015IntroductionQuantumStateEstimation}, share a technical backbone: they all involve linear post-processing of measurement data to infer properties of an unknown quantum state.
These schemes differ, however, in how the estimator is learned, what prior information is assumed, and the metrics used to assess performance.
Of particular interest are single-setting measurement strategies, which employ a fixed measurement device capable, in principle, of providing complete information about the measured states~\cite{renes2004Symmetric,allahverdyan2004Determining,bent2015Experimental,oren2017Quantum,stricker2022Experimental,zhao2025Efficient}.
These schemes notably avoid the need to reconfigure the measurement apparatus across exponentially many settings, thus offering significant experimental and computational advantages.
In all these cases, the dimension of the operator span of the POVM describing the overall measurement is what determines which features of the input states are retrievable. Nevertheless, beyond this abstract characterization, the specific quantum resources underlying the metrological power of these strategies has remained elusive.

In the context of quantum computation, a key resource that quantifies both simulation hardness and computational universality is \textit{non-stabilizerness} --- that is, a circuit’s ability to generate \textit{quantum magic}%
~\cite{leone2021renyi,Veitch2014resource,howard2017application,Wang2019magic,oliviero2022a,bravyi2005universal,howard2014ContextualitySuppliesMagic,oliviero2022Measuring}.
Prior work has analysed the structure of channels built with stabilizer resources~\cite{yashin2025Characterization}, efficient process-tomography methods for Clifford circuits~\cite{xue2023Efficient}, and the complexity of determining whether a state is a stabilizer~\cite{gross2021Schur}.
Other studies have proposed methods to build informationally-complete POVMs (IC-POVMs) from stabilizer and magic states~\cite{planat2017MagicInformationallyComplete,feng2022Stabilizer}, analysed the CNOT-cost of implementing IC-POVMs~\cite{you2025Circuit}, and studied the simulability of POVMs in terms of projective measurements without ancillas~\cite{oszmaniec2019Simulating}.
The generation of state designs with doped Clifford circuits and different Clifford orbits has also been analysed~\cite{zhang2025DesignsMagicaugmentedClifford,gross2021Schur}.
Yet, the metrological role of quantum magic remains largely unexplored. In particular, there is no general characterisation of the structure and completeness of POVMs generated by Clifford and magic-doped circuits.

In this work, we bridge this gap by showing the relation between non-stabilizerness budget --- quantified as the number of $T$ gates in the circuit~\cite{bravyi2005universal,leone2021renyi,jiang2023Lower,leone2024Stabilizer,ahmadi2024Mutual,haug2025Probing} --- and the amount of information retrievable from the corresponding measurement.
As said, we focus on single-setting scenarios, where the quantum circuit preceding the measurement is fixed, and no classical randomness is involved.
In particular, we achieve the following:
\begin{enumerate}
    \item We prove that any circuit that uses only Clifford gates and stabilizer ancillas is informationally equivalent to a direct projective measurement and thus cannot be IC.
    \item We prove that achieving informational completeness requires at least $2n/\log_23\approx 1.26 n\,$ $T$ gates, and that $2n$ suffice, where $n$ is the number of input qubits. We conjecture that $2n$ is in fact also the number of $T$ gates necessary for informational completeness.
    \item We prove a direct link between the entanglement of the Heisenberg-evolved stabilizer measurement states and the informational content of the resulting measurement.
\end{enumerate}
These findings show an interesting departure from the standard computational narrative surrounding quantum magic:
Clifford evolutions, though powerful for state manipulation, never yield more information than a straightforward computational‑basis read‑out;
conversely, whereas universal quantum computation demands an unbounded supply of $T$ gates, $2n$ such gates are already enough to guarantee informational completeness.

Another immediate application arises in the context of QELMs --- namely, quantum machine learning protocols that trade assumed knowledge of the dynamics for knowledge of a set of training states.
Our results tie the performance of QELMs to the non-stabilizerness of the underlying dynamics, thus allowing to better devise suitable reservoirs for the task one needs to solve.
Finally, from the perspective of shadow tomography, our results imply that stabilizer resources alone cannot be used to retrieve arbitrary observables on input states.
Note that this does not contradict standard results about the feasibility of shadow tomography with random Clifford circuits~\cite{huang2020predicting}, as these rely on measurements in multiple bases, obtained after evolution through many different Clifford circuits; by contrast, we consider POVMs obtained with only one such circuit.
However, our results might also help devising more efficient shadow tomography protocols, especially given the many recent efforts devoted to gain a better understanding of the number of Clifford circuits and the non-stabilizerness required to achieve high-quality performances with these schemes~\cite{haferkamp2023EfficientUnitaryDesigns,helsen2023ThriftyShadowEstimation,bertoni2024ShallowShadowsExpectation,bu2024ClassicalShadowsPauliinvariant,zhang2024MinimalCliffordShadow}.

The remainder of this paper is organized as follows.
\Cref{sec:summary_results} presents the technical results of the paper and serves as a reading guide, providing pointers to all sections and theorems.
\Cref{sec:setting} introduces the general framework and notation used throughout.
\Cref{sec:clifford_povms} analyses the estimation capabilities of POVMs built solely from stabilizer resources.
\Cref{sec:entanglement_and_reconstruction} discusses the role of entanglement in the measurement basis for general choices of initial ancillas.
\Cref{sec:tdoped_POVMs} investigates how these capabilities change when $T$ gates are added to the Clifford circuit.
Finally,~\cref{sec:conclusions} summarizes our findings and outlines possible venues for future work.


\section{Detailed summary of results}
\label{sec:summary_results}

We study measurements on $n$-qubit systems obtained by appending $m$ ancillary qubits, applying a unitary $U$ to the joint system, and then performing a projective measurement. This construction defines a POVM $\bs\mu=(\mu_b)_{b=1}^{\nout}$ with $\nout=2^{n+m}$ outcomes, whose explicit form is given in~\cref{eq:effective_POVM}.
Our main object of interest is the \textit{information content} of such POVMs, as quantified by the dimension of the linear operator span of their elements, $\dimspanmu\equiv\dim\on{span}(\{\mu_b\}_{b=1}^{\nout})$.
This equals the number of linearly independent observables reconstructable exactly from measurement data via linear post-processing.

For generic choices of unitaries, measurements, and ancilla states, such isometric extensions yield POVMs with larger $\dimspanmu$.
In particular, adding $n$ ancillas is both necessary and sufficient --- given suitable $U$ --- to obtain informationally complete (IC) POVMs.
In~\cref{sec:clifford_povms} we show that the opposite is true if only stabilizer resources are used. Specifically, in~\cref{thm:structural_clifford_povms,thm:structural_clifford_povms_plus} we prove:
\begin{result}[Stabiliser POVMs]
    Effective POVMs built solely from stabilizer resources are informationally equivalent to projective measurements in a stabilizer basis. They satisfy $\dimspanmu=2^n$, independently of the number of ancillas, and in particular are never IC.
\end{result}
The proof exploits the structure of the stabilizer groups characterizing both measurement and ancillas.
Indeed, a stabilizer measurement is defined as the projective measurement whose elements are the common eigenvector of some stabilizer group.
In the Heisenberg picture, Clifford evolutions map the stabilizer measurement into another stabilizer measurement, characterized by a new stabilizer group.
The stabilizer ancilla states, also characterized by some other stabilizer group, appear in this picture as a projection applied to the Heisenberg-evolved stabilizer states defining the measurement.
This allows to characterize the resulting effective POVMs, and explicitly derive the effective measurements from the defining stabilizer groups.

We then broaden the scope of our analysis in~\cref{sec:entanglement_and_reconstruction} to stabilizer measurements and Clifford circuits allowing now  arbitrary ancilla states.
In the Heisenberg picture, it is clear that the entanglement of the states corresponding to the POVM elements is determinant to the informational power of the measurement.
For example, if the measurement operators remain fully factorized, only local observables are accessible, and thus the POVM is not IC.
Leveraging the entanglement characterization of stabilizer states from~\cite{fattal2004entanglement}, we prove in~\cref{thm:entanglement_and_reconstruction} that this intuition can be made precise:
\begin{result}[Entanglement and informational power]
    For stabilizer measurements, Clifford circuits, and arbitrary ancillas, the effective POVM satisfies $2^n \le \dimspanmu \le 2^{n+p}$, where $p$ is the entanglement entropy of the Heisenberg-evolved measurement states.
    For almost all ancilla choices, $\dimspanmu=2^{n+p}$; special cases (e.g. stabilizer ancillas) may yield smaller values.
\end{result}
We further argue that the maximum $\dimspanmu$ is always achieved when the Heisenberg-evolved measurement basis is maximally entangled, independent of the ancillas state.
Note how this differs from the previous result: while $\dimspanmu\le 2^{n+p}$ holds for all ancilla states, it remains possible that for specific choices the maximum occurs at some $p<n$, and furthermore that for those ancillas $p=n$ yields smaller values of $\dimspanmu$.
Numerical evidence shows no such cases, and their existence would be surprising. The proof of this result is left as a \hyperref[conj:maxentmaxrank]{conjecture}.

In~\cref{sec:tdoped_POVMs} we drop the Clifford restriction and consider circuits with $t$ $T$ gates.
These are significantly harder to analyze since the Heisenberg-evolved POVM elements are no longer stabilizer states.
To address this, we exploit the gadget formalism: each $T$ gate is replaced by an ancilla, a CNOT, and a projection onto $\ket0$.
Because CNOTs are Clifford, this reduces the non-Clifford circuit to a Clifford one with extra ancillas and projections. The stabilizer formalism then applies again, provided we suitably track the projected ancillas.
With this method, together with a convenient way to characterize $\dimspanmu$ via specific subgroups of the stabilizer measurement group, we prove in~\cref{thm:easy_tn_bound,thm:tgen} that:
\begin{result}[$t$-doped POVMs]
    An effective POVM can be IC only if $t\ge 2n/\log_2 3$.
    Moreover, there exist circuits with $t=2n$ yielding IC POVMs, so $t\ge 2n$ is a sufficient condition.
\end{result}
Numerical evidence suggests $t=2n$ is also necessary: no $t$-doped POVM with fewer than $2n$ $T$ gates appears to be IC. This stronger claim is left as an unproven \hyperref[conj:2n]{conjecture}.
We also prove in~\cref{thm:rank_tlen} that for maximally entangled measurements with $t\le n$, the largest achievable $\dimspanmu$ among all $t$-doped POVMs is $\max\dimspanmu=2^n (3/2)^t$, whereas for $t>n$ we prove in~\cref{thm:rank_tgen} that $\max\dimspanmu\ge 2^{-\ell}(3^{a+1}-1)^r(3^a-1)^{\ell-r}$, with $a\equiv\lfloor t/\ell]$, $r\equiv t-a\ell$, and $\ell\equiv t-n$.

\section{Setting and notation}
\label{sec:setting}

\parTitle{Physical and effective POVMs}
Our aim is to characterize how much information can be extracted about an $n$-qubit input state $\rho$ with measurement protocols consisting of
(1) embedding $\rho$ into a larger Hilbert space by appending $m$ ancillary qubits;
(2) evolving the overall state by a quantum channel $\Phi$;
(3) measuring the output state with a projective POVM $\bs\mu^{\rm phys}\equiv (\mu_b^{\rm phys})_{b=1}^{n_{\rm out}}$, where $\nout=2^{n+m}$ is the number of outcomes.
We will refer to the $n$-qubit register holding the input state to measure as the \textit{data qubits}.
The Hilbert spaces of data qubits and ancillas are denoted with $\calH_n$ and $\calH_m$, respectively.
If the only quantity of interest is the information that can be recovered about $\rho$, then the distinction between the initial ancilla state, the dynamics, and the physical measurement, becomes moot. Operationally, all that matters is the \textit{effective measurement} $\bs\mu\equiv(\mu_b)_{b=1}^{n_{\rm out}}$.
This is the POVM in the Heisenberg picture, given by $\mu_b = \Phi^\dagger(\mu_b^{\rm phys})$, with $\Phi^\dagger$ the adjoint map of $\Phi$.
The main quantity we are interested in is the dimension of the linear span of the elements of the effective POVM.
We shall denote this quantity with $\dimspanmu\equiv\dim\on{span}(\{\mu_b\}_{b=1}^{\nout})$.
This equals the number of linearly independent observables that can be written as linear combinations of the measurement probabilities given by $\bs\mu$.
Informationally-complete (IC) POVMs are thus characterised by $\dimspanmu=2^{2n}$.
A schematic representation of the considered measurement protocols is given in~\cref{fig:apparatus}.

\begin{figure}[tb]
    \centering
    \includegraphics[width=0.9\linewidth]{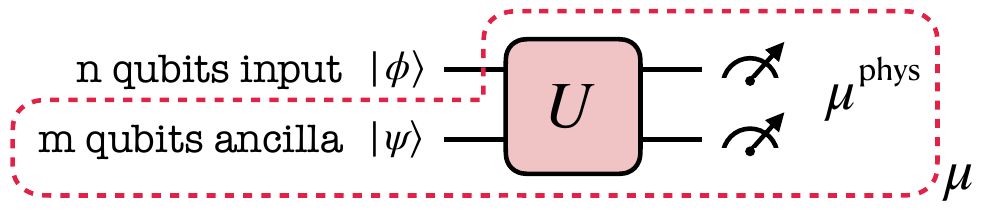}
    \caption{Schematic representation of the types of measurements considered in the paper, with $\Phi$ assumed to be a unitary evolutino $U$.}
    \label{fig:apparatus}
\end{figure}

\parTitle{Explicit form of the effective POVMs}
We specialize in particular to channels of the form $\Phi(\rho) = U(\rho\otimes\rho_R)U^\dagger$, with ancillas $\rho_R=\mathbb{P}_\psi\equiv\ketbra\psi$ initialized in a fixed pure states $\ket\psi$, and $U$ as a unitary transformation acting on the joint $(n+m)$-qubit system.
The physical measurement is assumed to be projective in an orthonormal basis $\ket*{\Phi_b'}$, so that $\mu^{\rm phys}_b=\mathbb{P}_{\Phi_b'}$.
Under these assumptions, the effective POVM elements take the form
\begin{equation}\label{eq:effective_POVM}
\begin{gathered}
    \mu_b = \Trace_R[U^\dagger \mathbb{P}_{\Phi_b'} U (\mathbb I\otimes \rho_R)]
    \\= (I\otimes \bra\psi)\mathbb{P}_{\Phi_b}(I\otimes\ket\psi),
\end{gathered}
\end{equation}
where the partial trace is taken over the ancillary qubits, and we denoted with $\mathbb{P}_{\Phi_b}\equiv\ketbra{\Phi_b}$, $\ket{\Phi_b}\equiv U^\dagger \ket{\Phi_b'}$ the Heisenberg-evolved measurement basis elements in the enlarged space.
We call POVMs of the form in~\cref{eq:effective_POVM} \emph{stabilizer POVMs} when they use only stabilizer resources — that is, when $U$ is a Clifford operator, $|\psi\rangle$ is a stabilizer state, and $\{|\Phi_b\rangle\}_b$ is a stabilizer basis.
Note that if $U$ is Clifford, then $\{\ket{\Phi_b}\}_b$ is a stabilizer basis iff $\{\ket{\Phi_b'}\}_b$ is.
If instead $U$ is $t$-doped, meaning that it contains $t$ $T$ gates, we will refer to the resulting POVMs as \textit{$t$-doped POVMs}.
We still assume stabilizer ancillas and stabilizer physical POVM for $t$-doped POVMs, because any non-stabilizerness can be reabsorbed into $U$, and thus there is no loss of generality in studying the general features of~\cref{eq:effective_POVM} where $U$ is the only non-stabilizer component.

\parTitle{Stabilizers formalism}
Here we provide a brief overview of the main concepts and notation related to the formalism of stabilizers states that will be useful in this paper.
Let $\calP_1$ be the {Pauli group}, that is the group of single-qubit Pauli operators with coefficients taken in $\{\pm1,\pm i\}$, so that $\lvert\calP_1\rvert=16$.
We will denote with $\calP_n\equiv \calP_1^{\otimes n}$ the corresponding group of $n$-qubit Pauli operators.
In many applications, the phases attached to each Pauli operator are immaterial, so that only the structure of $\calP_n$, modulo such phases, is relevant. We then work with the group of \textit{Pauli strings} $\tildecalP_n$, defined as the quotient of $\calP_n$ over the phases: $\tildecalP_n\equiv \calP_n/\{\pm1,\pm i\}$. We will denote the elements of $\calP_n$ using a standard font, \textit{e.g.} $X\in\calP_1$, $YZ\in\calP_2$, and use instead a monospaced font to denote the elements of $\tildecalP_n$, e.g. $\ps{I}\in\tildecalP_1$, $\ps{XX}\in\tildecalP_2$, etc. (here $\ps{I}$ and $I$ denote the identity operator, while $\ps{K}$ and $K$ the $k$-Pauli matrix for $k=x,y,z$).
With this notation, we thus have $X\cdot Y=iZ$ but $\ps{X}\cdot \ps{Y}=\ps{Z}$.

A \textit{stabilizer group} over $n$ qubits is defined as a maximal Abelian subgroup $\calS\le \calP_n$ such that, for some state $\ket\psi$, we have $g\ket\psi=\ket\psi,~\forall g\in\calS$.
Equivalently, a stabilizer group can be defined as a maximal Abelian subgroup of $\calP_n$ that does not contain $-I$.
A \textit{stabilizer state} is a pure state that is stabilised by some stabilizer group.
We will instead talk of a \textit{stabilizer basis} to refer to an orthonormal basis of pure (stabilizer) states that are the common eigenvectors of some stabilizer group.
Because in such definition the phases attached to the elements of a stabilizer group $\calS\le\calP_n$ are not relevant, we will associate stabilizer bases to \textit{stabilizer groups modulo phases}, that is, to $\calS\le\tildecalP_n$.
So for example, the Bell state $\ket{00}+\ket{11}$ is the stabilizer state corresponding to the stabilizer group $\calS=\langle XX,ZZ\rangle=\{I,XX,-YY,ZZ\}\le\calP_2$, while the set of four Bell states is the stabilizer basis associated to $\calS=\langle\ps{XX},\ps{ZZ}\rangle=\{\ps{I},\ps{XX},\ps{YY},\ps{ZZ}\}\le\tildecalP_2$.
We will refer to both $\calS\le\calP_n$ and $\calS\le\tildecalP_n$ as ``stabilizer groups'', specifying which notion we are referring to when needed.
Sporadically, we will simplify the notation for the sake of ease of argument.
For instance, given $g\in\tildecalP_n$, we will write $I\pm g$ where it is clear that these are to be interpreted as proper operators $I\pm g\in\calP_n$, although we should more formally denote such objects as $I\pm \pi(g)\in\calP_n$ with $\pi:\tildecalP_n\to\calP_n$ the natural lifting from the quotient space to the space of actual Pauli operators.
Generally speaking, whenever a sign or sum is applied to an element of $\tildecalP_n$, we assume that such natural lifting has been employed.
We refer to~\cref{app:stab} for a more detailed review of the necessary background on the stabilizer formalism.

\section{Reconstruction with stabilizer POVMs}
\label{sec:clifford_povms}

In this Section, we prove that every stabilizer POVM is informationally equivalent to a projective measurement in a stabilizer basis.
The main structural result is~\cref{thm:structural_clifford_povms}, which gives a short argument for this equivalence.
We then strengthen the result through~\cref{thm:structural_clifford_povms_plus}, using a different approach that makes the structure of the effective POVMs more explicit and provides a concrete recipe to compute them for any chosen stabilizer basis and projection.
These results imply, in particular, that stabilizer POVMs are never IC, and always give $\dimspanmu=2^n$.
This contrasts with the case of general $U$ and $\ket\psi$: POVMs of the form in~\cref{eq:effective_POVM} can generally --- and almost always with random $U$ --- have $\dimspanmu=\min(2^{n+m},4^n)$, and thus become IC for $m\ge n$.
Our result shows that, with stabilizer resources, the opposite holds: $\dimspanmu$ does not depend on the number of ancillas.

We first show in~\cref{lemma1} that local projections cannot \textit{decrease} the information retrievable from a rank-1 POVM.
Physically, this means that adding ancillas, evolving the full system unitarily, and then performing a projective measurement, provide at least as much information about the system as a direct projective measurement on it, \textit{i.e.} $\dimspanmu\ge 2^n$.
This tells us that it is not possible to have too much information hiding in unobserved correlations between the output qubits.
\begin{lemma}\label{lemma1}
    Let $\bs\mu\equiv(\mu_b)_{b=1}^{dd'}$ be a rank-1 POVM of the form $\mu_b = (I\otimes\bra\psi)\ketbra{\Phi_b}(I\otimes\ket\psi)$ with $\ket{\Phi_b}\in\mathbb{C}^{d\times d'}$ some orthonormal basis, and $\ket\psi\in\mathbb{C}^{d'}$ some ancillary state that is projected before the measurement.
    Then $\dimspanmu\ge d$.
\end{lemma}
\begin{proof}
    Suppose by contradiction that $r<d$, and let $\ket{\phi_b}\propto (I\otimes \bra\psi)\ket{\Phi_b}$.
    There must then be some $\ket{\phi_\perp}$ such that $\langle \mathbb{P}_{\phi_\perp},\mathbb{P}_{\phi_b}\rangle=0$ for all $b$, or equivalently, such that $\langle\phi_\perp|\phi_b\rangle=0$ for all $b$.
    But, then we could build the vector $\ket{\Phi'}\equiv \ket{\phi_\perp}\otimes \ket\psi$, and observe that this satisfies
    \begin{equation}
        \braket{\Phi'}{\Phi_b}
        = \sqrt{w_b}\langle \phi_\perp| \phi_b\rangle
        = 0, \,\,\forall b.
    \end{equation}
    But this is a contradiction on account of $\ket{\Phi_b}$ being a basis.
\end{proof}

Let us now focus on Clifford evolutions with stabilizer measurements and ancillas. In this case, we find that adding ancillas also does not \textit{increase} $\dimspanmu$:
\begin{theorem}\label{thm:structural_clifford_povms}
Any stabilizer POVM is informationally equivalent to a projective measurement in an orthonormal basis of stabilizer states --- and thus has $\dimspanmu=2^n$.
\end{theorem}
\begin{proof}
    We consider POVMs of the form given in ~\cref{eq:effective_POVM} with $\{\ket{\Phi_b}\}_b$ a stabilizer measurement basis over $n+m$ qubits, and $\ket\psi$ an $m$-qubit stabilizer state.
    Let the stabilizer group characterizing $\{\ket{\Phi_b}\}_b$ be $\calS=\langle g_1,...,g_{n+m}\rangle\subset\tildecalP_{n+m}$, with $g_i$ independent generators, and let $\calZ\le \calP_m$ be the stabilizer group of $\ket\psi$.
    Here and in the following, we denote with $\calH_n$ the $n$-qubit Hilbert space of the data qubits, and with $\calH_m$ the $m$-qubit Hilbert space of the ancillas.

    Define the modified operators $\mu_g'\equiv \Trace_2[g(I\otimes\mathbb{P}_\psi)]$, and note that there is a linear isomorphism between $(\mu_b)_{b\in\{0,1\}^{n+m}}$ and $(\mu_g')_{g\in \calS}$. This follows from the relation between stabilizer states and corresponding stabilizer group: we have $\mathbb{P}_{\Phi_b} = \prod_{k=1}^{n+m} \frac{I+(-1)^{b_k}g_k}{2}$, and therefore
    \begin{equation}
    \begin{gathered}
        \mathbb{P}_{\Phi_b}
        = \frac1{2^{n+m}} \sum_{\xi\in\{0,1\}^{n+m}}(-1)^{\xi\cdot b} g^\xi, \\
        g^\xi = \sum_{b\in\{0,1\}^{n+m}} (-1)^{\xi\cdot b} \mathbb{P}_{\Phi_b},
    \end{gathered}
    \end{equation}
    where $\xi\cdot b\equiv\sum_{k=1}^{n+m}\xi_k b_k\pmod2$ is the standard scalar product in $\mathbb{Z}_2^{n+m}$, and we used the shorthand $g^\xi\equiv \prod_{k=1}^{n+m} g_k^{\xi_k}$ to index the elements of $\calS$.
    This map is the group Fourier transform of the stabilizer group elements.
    The same identical coefficients are used to map between $(\mu_b)$ and $(\mu_g')$ and vice versa. Namely,
    \begin{equation}\label{eq:phys_vs_prime_POVMs}
    \begin{gathered}
        \mu_b=2^{-n-m}\sum_\xi (-1)^{\xi\cdot b}\mu_{g^\xi}',
        \\
        \mu_{g^\xi}'=\sum_b (-1)^{\xi\cdot b}\mu_b.
    \end{gathered}
    \end{equation}
    In particular, both sets span the same space: $\on{span}(\{\mu_b\})=\on{span}(\{\mu_g^\prime\})$.

    Writing each $g\in\calS$ as $g=\pi_n(g)\otimes \pi_m(g)$, with $\pi_n(g)$, $\pi_m(g)$ the projections onto $\calH_n$ and $\calH_m$, respectively, we have $\mu_g'=
    \pi_n(g)\langle\psi|\pi_m(g)|\psi\rangle$.
    This shows that $\mu_g'=\pi_n(g)$ iff $\pm \pi_m(g)\in\calZ$, and $\mu_g'=0$ otherwise.
    The problem of determining the information retrievable from the measurement is thus reduced to that of determining the dimension of the subgroup
    \begin{equation}\label{eq:gn_for_ginS}
        \calS_\calZ\equiv \{\pi_n(g)\in\tildecalP_n: \, g\in\calS,\,\,\pm \pi_m(g)\in\calZ\}.
    \end{equation}
    For any $g_1,g_2\in\calS$, $\pi_n(g_1)$ and $\pi_n(g_2)$ (anti) commute iff $\pi_m(g_1)$ and $\pi_m(g_2)$ do.
    It follows that $\calS_\calZ$ must be an abelian subgroup of $\tildecalP_n$, and thus have $\dimspanmu\le 2^n$.
    But~\cref{lemma1} implies $\dimspanmu\ge 2^n$, hence we conclude that $\dimspanmu=2^n$.
    As the elements $\mu_g'$ are all Pauli operators in $\tildecalP_n$, the effective measurement must be informationally equivalent to a projective measurement on some $n$-qubit stabilizer basis.
\end{proof}
Note that~\cref{thm:structural_clifford_povms} does not imply that $\bs\mu$ itself must be a projective measurement on $\calH_n$: there are precisely $2^n$ distinct nonzero operators $\mu_g'$, but nonetheless the number of nonzero elements $\mu_b$ can be larger.
In fact, $\bs\mu$ cannot generally be a projective measurement, because it has $2^{n+m}>2^n$ outcomes.
The equivalence of the POVMs is thus understood here in the sense of~\cite{zhu2022Quantum}: two POVMs are said to be equivalent when they provide exactly the same informational content, despite their elements not necessarily being identical as operators. This is akin to how the three-outcome single-qubit POVM $\{\frac12\mathbb{P}_0,\frac12\mathbb{P}_0,\mathbb{P}_1\}$ is \textit{informationally equivalent} to a $Z$-basis measurement, despite not being a projective measurement \textit{per se}.
    
While the above reasoning is sufficient to prove the result, it is instructive to consider a different approach to the proof that does not explicitly involve~\cref{lemma1}, exposes the structure of the effective POVM, and sheds more light on the origin of the stated equivalence.
To set up the stage, we must first state the following:

\begin{lemma}\label{lemma:intersection_S_and_Z}
    Let $\calS\le\tildecalP_{\nu}$ be a maximal abelian subgroup, and let $\calZ\le\tildecalP_{\nu}$ be another abelian subgroup, with $\dim(\calS)=\nu$ and $\dim(\calZ)=m\le \nu$.
    Then we can find independent generators for $\calS$ and $\calZ$ such that
    \begin{equation}
    \begin{gathered}
        \calZ = \langle \{h_k\}_{k=1}^\ell\cup \{\tilde h_k\}_{k=1}^{m-\ell} \rangle, \\
        \calS = \langle
        \{h_k\}_{k=1}^\ell \cup \{\tilde g_k\}_{k=1}^{m-\ell}
        \cup \{ g_k\}_{k=1}^{\nu-m}
        \rangle,
    \end{gathered}
    \end{equation}
    where $\langle h_1,\dots,h_\ell\rangle=\calS\cap\calZ$, $\langle h_1,\dots,h_\ell,g_1,\dots, g_{\nu-m}\rangle=\calS\cap C(\calZ)$, and $\{ \tilde h_j,\tilde g_k\}=0$ iff $j=k$. Here $C(\calZ)$ is the centralizer of $\calZ$, that is, the set of $P\in \tildecalP_\nu$ such that $[P,\calZ]=0$.
\end{lemma}
\begin{proof}
    This statement can be seen as a direct consequence of Theorem 1 in Ref.~\cite{fattal2004entanglement}, taking as group the free product $\calS*\calZ$, whose maximal Abelian subgroup is $\calS$, and whose center is $\calS\cap C(\calZ)$. For completeness, we nonetheless provide a full tailor-made proof here.

    Observe that $\calS$ contains the chain of subgroups $\calS\cap\calZ\le \calS\cap C(\calZ)\le \cal S$, and thus we have an isomorphism
    \begin{equation}
    \begin{gathered}
        \calS \simeq \calS_\calZ \times (\calS_{C(\calZ)}/\calS_\calZ)
        \times (\calS/\calS_{C(\calZ)}), \\
        \calS_\calZ\equiv \calS\cap\calZ,
        \quad
        \calS_{C(\calZ)}\equiv \calS\cap C(\calZ),
    \end{gathered}
    \end{equation}
    where $(\cdot)\times(\cdot)$ denotes the external product of groups.
    More explicitly in terms of the generators, given $\dim(\calS/(\calS\cap C(\calZ)))=s$, we can find independent generators for $\calS$ divided as
    \begin{equation}
        \{h_k\}_{k=1}^\ell \cup \{\tilde g_k\}_{k=1}^{s} \cup \{ g_k\}_{k=1}^{\nu-\ell-s},
    \end{equation}
    where $\{h_k\}$ spans $\calS\cap\calZ$, $\{ h_k\}\cup \{ g_k\}$ spans $\calS\cap C(\calZ)$ with $g_k\in \calS\cap C(\calZ)\setminus \calZ$, and $\{ h_k\}\cup \{ g_k\}\cup \{\tilde g_k\}$ span $\calS$ with $\tilde g_k\in\calS\setminus C(\calZ)$.
    Explicitly, this is done taking $g_k$ as representative of a set of independent generators $[g_k]\equiv g_k(\calS\cap \calZ)$ for the quotient space $ (\calS\cap C(\calZ))/(\calS\cap \calZ)$, and $\tilde g_k$ as representatives for a set of independent generators $[\tilde g_k]\equiv \tilde g_k(\calS\cap C(\calZ))$ for the quotient space $ \calS/(\calS\cap C(\calZ))$. Similarly, $\calS\cap\calZ\le \calZ$ and thus we can find independent generators for $\calZ$ of the form $\{ h_k\}_{k=1}^\ell\cup \{\tilde h_k\}_{k=1}^{m-\ell}$, with $\tilde h_k\in\calZ\setminus \calS$.

    To prove the Lemma it remains to show that $s=m-\ell$, and that $\{\tilde h_k\}_{k=1}^{m-\ell}$ and $\{ \tilde g_k\}_{k=1}^{m-\ell}$ can be given the ``diagonal anticommutation pattern'' as per the statement.
    To this end, we can follow an iterative procedure to take any such pair of generators, and modify them to ensure they satisfy the given properties:
    \begin{enumerate}
        \item Start with any $\tilde h_1\in\calZ\setminus\calS$ and note that there is at least one $\tilde g_k$ such that $\{\tilde h_1,\tilde g_k\}=0$. If there is more than one, change the basis multiplying the anticommuting ones together so that only a single generators, call it $\tilde g_1$, is left to anticommute with $\tilde h_1$.
        \item Take the next generator $\tilde h_2$. If $\{\tilde h_2,\tilde g_1\}=0$, then replace $\tilde h_2\to \tilde h_1\tilde h_2$, so that the new $\tilde h_2$ commutes with $\tilde g_1$. We now proceed as in the previous step looking for a unique generator in the subset $\{\tilde g_2,\dots,\tilde g_s\}$ that anticommutes with $\tilde h_2$.
        \item Repeat the above process for all $m-\ell$ generators $\tilde h_k$, at each step first ensuring that $\tilde h_k$ commutes with all $\tilde g_1,\dots,\tilde g_{k-1}$, and then ensures it anticommmutes with a single other generator denoted $\tilde g_k$.
    \end{enumerate}
    Note that in the above procedure there must always be a new generator $\tilde g_k$ that anticommutes with $\tilde h_k$, as otherwise we would have $\tilde h_k\in\calS\cap\calZ$.
    And furthermore after $m-\ell$ steps there cannot be $\tilde g_k$ generators left out, as that would mean that $\tilde g_k\in \calS \cap C(\calZ)$.
    Thus $s=m-\ell$ and $\{\tilde h_j,\tilde g_k\}=0 \iff j=k$.
\end{proof}

\begin{figure}[tbp]
    \centering
    \includegraphics[width=0.5\linewidth]{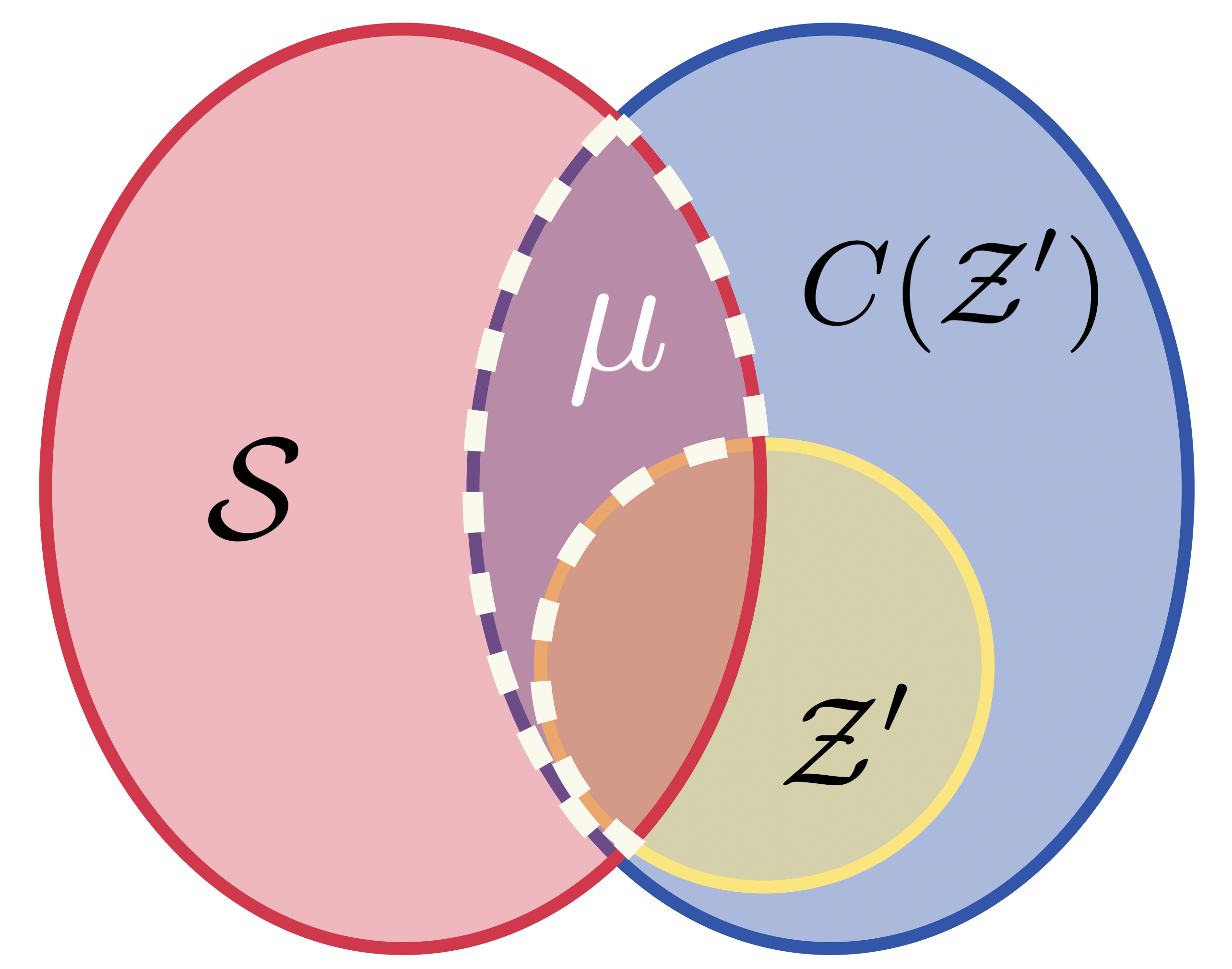}
    \caption{Visual representation of the result of~\cref{thm:structural_clifford_povms_plus}. The generators that determine the effective POVM $\bs\mu$ are elements of $(\calS\cap C(\calZ'))\setminus \calZ'$. More precisely, they are a set of nontrivial representatives for the quotient space $(\calS\cap C(\calZ')) / (\calS\cap \calZ')$.}
    \label{fig:scheme_thm2_generators_partition}
\end{figure}

\Cref{lemma:intersection_S_and_Z} makes the characterization of the effective measurement relatively straightforward. The gist is that the effective measurement is determined by $(\calS\cap C(\calZ))/(\calS\cap \calZ)$, that is by the elements of $\calS$ which commute with $\calZ$ but are not in it, and that projecting these elements onto $\calH_n$ always gives precisely $n$ independent generators, which describe the measured directions:
\begin{theorem}\label{thm:structural_clifford_povms_plus}
    Under the same conditions of~\cref{thm:structural_clifford_povms}, the effective POVM has $2^{n+m-\ell}$ nonzero elements, for some $0\le \ell\le m$.
    The nonzero elements have the form $\mu_b = 2^{\ell-m}\mathbb{P}_{\Psi_b}$, with $\{\ket{\Psi_b}\}$ a stabilizer basis on $\calH_n$. Furthermore, each of these $2^n$ distinct elements is repeated $2^{m-\ell}$ times in the POVM.
\end{theorem}
\begin{proof}
    Applying~\cref{lemma:intersection_S_and_Z} to $\calS\le\tildecalP_{n+m}$ and $\calZ'\equiv I\otimes\calZ$, we obtain generators of the form
    \begin{equation}\label{eq:goodgenerators_S}
    \begin{gathered}
        \calZ' = \langle\{ I\otimes h_k\}_{k=1}^\ell\cup \{I\otimes \tilde h_k\}_{k=1}^{m-\ell}\rangle, \\
        \calS = \langle
        \{ I\otimes h_k\}_{k=1}^\ell \cup
        \{\tilde g_k\}_{k=1}^{m-\ell} 
        \cup \{ g_k\}_{k=1}^{n}
        \rangle,
    \end{gathered}
    \end{equation}
    such that $\{I\otimes h_k\}$ generate $\calS\cap\calZ'$, $\{I\otimes h_k\}\cup \{g_k\}$ generate $\calS\cap C(\calZ')$, and all generators commute except for $\{\tilde g_k,I\otimes \tilde h_k\}=0$ for all $k$.
    To make the intersections $\calS\cap\calZ'$ and $\calS\cap C(\calZ')$ well-defined, we view here $\calZ'$ as a subgroup of $\tildecalP_{n+m}$, that is, we ignore the signs in its elements.
    
    For any $g\in\calS$ and any $\ket\psi$ stabilised by $\calZ$, write
    \begin{equation}
    \begin{gathered}
        g=\prod_{k=1}^\ell (I\otimes h_k)^{\alpha_k}\prod_{k=1}^{m-\ell} \tilde g_k^{\beta_k}\prod_{k=1}^{n} g_k^{\gamma_k},\\
        \mathbb{P}_\psi =
        \prod_{k=1}^\ell \frac{I+(-1)^{d_k} h_k}2
        \prod_{k=1}^{m-\ell} \frac{I+(-1)^{e_k} \tilde h_k}2,
    \end{gathered}
    \end{equation}
    with $\alpha_k,\beta_k,\gamma_k,e_k,d_k\in\{0,1\}$.
    Then $\mu'_g=\pi_n(g)\langle\psi|\pi_m(g)|\psi\rangle$ becomes:
    \begin{equation}
    \begin{gathered}
        \mu_g' =
        A_{\psi,g}\prod_{k=1}^{m-\ell}\pi_n(\tilde g_k)^{\beta_k}
        \prod_{k=1}^{n}\pi_n(g_k)^{\gamma_k},
        \\
        A_{\psi,g}\equiv \langle\psi|
        \prod_{k=1}^{\ell}h_k^{\alpha_k}
        \prod_{k=1}^{m-\ell}\pi_m(\tilde g_k)^{\beta_k}
        \prod_{k=1}^{n}\pi_m(g_k)^{\gamma_k}
        |\psi\rangle.
    \end{gathered}
    \end{equation}
    By construction, the generators in~\cref{eq:goodgenerators_S} satisfy $[h_i,\pi_m(\tilde g_j)]=[h_i,\pi_m(g_j)]=[h_i,\tilde h_j]=0$ for all $i,j$, while $\{\tilde h_k,\pi_m(\tilde g_k)\}=0$ and $\tilde h_k$ commutes with all other generators.
    Hence $A_{\psi,g}=0$ unless $\beta_k=0$ for all $k$.
    If $b=0$, then $A_{\psi,g}=\pm1$ since the operator in the expectation value is a Pauli operator that commutes with the stabilizer group $\calZ$ fo $\ket\psi$.
    Therefore\begin{equation}\label{eq:explicit_mugprime}
        \mu_g' = \pm \delta_{\beta,0}
        \prod_{k=1}^{n} \pi_n(g_k)^{\gamma_k},
    \end{equation}
    using the shorthand $\delta_{\beta,0}\equiv\prod_{k=1}^{m-\ell}\delta_{\beta_k,0}$.
    Furthermore, the Paulis $\{\pi_n(g_k)\}_{k=1}^n$ are independent generators on $\calH_n$: if we had $\prod_k \pi_n(g_k)^{\gamma_k}=I$ with $\gamma_k$ not all zero, then the corresponding $\prod_k g_k^{\gamma_k}\in\calS\cap C(\calZ')$ would act trivially on $\calH_n$, hence be an operator commuting with $\calZ$ but not being in it, contradicting $\calZ$ being maximal abelian.
    We conclude that $\{\mu_g'\}_{g\in\calS}$ contains precisely $2^n$ distinct Pauli operators, hence $\dimspanmu=2^n$.
    A schematic visual representation of which generators survive and contribute to the effective POVM is given in~\cref{fig:scheme_thm2_generators_partition}.

    Going further, we observe that
    \begin{equation}
    \begin{gathered}
        A_{\psi,g} = \delta_{\beta,0}
        (-1)^{d\cdot\alpha}
        \langle\psi|\prod_{k=1}^n \pi_m(g_k)^{\gamma_k}|\psi\rangle \\
        = \delta_{\beta,0}(-1)^{d\cdot\alpha + d'\cdot\gamma},
    \end{gathered}
    \end{equation}
    where $d\cdot\alpha\equiv\sum_k d_k\alpha_k$, and $d'\in\{0,1\}^n$ is some binary vector that depends on the decomposition of $\pi_m(g_k)$ with respect to the generators of $\calZ$.

    We can now derive the explicit form of the effective POVM from $\mu_g'$ using~\cref{eq:phys_vs_prime_POVMs}; up to a relabeling of the indices we can write
    \begin{equation}
    \begin{gathered}
        \mu_{a,b,c}= 2^{-n-m}
        \sum_{\alpha,\beta,\gamma}
        (-1)^{a\cdot\alpha+b\cdot\beta+c\cdot\gamma}
        A_{\psi,g}
        \pi_n(g)^{\gamma}
        \\
        =2^{-n-m+\ell}
        \delta_{a+d,0}
        \sum_\gamma
        (-1)^{(c+d')\cdot\gamma}
        \pi_n(g)^\gamma.
    \end{gathered}
    \end{equation}
    We can now simply relabel the measurement outcomes so that $c+d'\to c$, $a+d\to a$, and observe that the resulting factor equals
    $2^{-n}\sum_\gamma(-1)^{c\cdot\gamma}\pi_n(g)^\gamma = \prod_{k=1}^n \frac{I+(-1)^{c_k}\pi_n(g_k)}{2}$, and conclude that
    \begin{equation}
        \mu_{a,b,c} = 2^{\ell-m}\delta_{a,0}
        \prod_{k=1}^n \frac{I+(-1)^{c_k} \pi_n(g_k)}{2}.
    \end{equation}
    In other words, we conclude that
    (i) the indices $b$ do not affect the operator, thus causing each distinct operator in the POVM to be repeated $2^{m-\ell}$ times;
    (ii) the outcomes corresponding to $a\neq0$ never occur, thus there are $2^{n+m-\ell}$ nonzero outcomes in total;
    (iii) the nonzero elements, up to the rescaling factor $2^{\ell-m}$, are precisely the elements of the stabilizer basis generated by $\{\pi_n(g_k)\}_{k=1}^n$.
    Note that in this expression there is a total of $2^{n+m}$ possible outcomes, $2^{n+m-\ell}$ of which survive the Dirac deltas, and $2^{m-\ell}$ of which are identical to each other.
    Thus summing over all $\bs b\in\{0,1\}^{n+m}$ correctly recovers the normalization condition.
\end{proof}

\hypertarget{ref:sec1}{Explicit example applications of~\cref{thm:structural_clifford_povms_plus} are given in~\cref{ex:thm2_1,ex:thm2_2}.}

\section{Entanglement and property reconstruction}
\label{sec:entanglement_and_reconstruction}

In this Section, we analyse the role of entanglement in the reconstruction capabilities of stabilizer POVMs.
Specifically, we study how the entanglement of the stabilizer basis $\{\ket{\Phi_b}\}_b$ in~\cref{eq:effective_POVM}, as quantified via the methods of Ref.~\cite{fattal2004entanglement}, relates to $\dimspanmu$.
The main result of this Section is~\cref{thm:entanglement_and_reconstruction}, which proves that using Clifford unitaries but allowing arbitrary ancilla states $\ket\psi$, the upper bound to the informational content of the measurement becomes $\dimspanmu\le 2^{n+p}$, where $p$ is the entanglement entropy of (any) $\ket{\Phi_b}$, which as shown in Ref.~\cite{fattal2004entanglement} is always an integer for stabilizer states.
The fact that these results hold for general $\ket\psi$ makes them a useful foundation for the analysis of $t$-doped POVMs, which will be discussed in~\cref{sec:tdoped_POVMs}.

\begin{theorem}\label{thm:entanglement_and_reconstruction}
    A POVM of the form of ~\cref{eq:effective_POVM}, with a Clifford $U$ and an arbitrary $\rho_R=\mathbb{P}_\psi$ has $\dimspanmu\le 2^{n+p}$, with $p$ the entanglement of the Heisenberg-evolved states $\ket{\Phi_b}$.
    In particular, IC-POVMs are possible iff the states are maximally entangled, \textit{i.e.} $m\ge n$ and $p=n$.
\end{theorem}
\begin{proof}
We consider POVMs of the form~\cref{eq:effective_POVM} with $\{\ket{\Phi_b}\}_b$ a stabilizer basis and $\ket\psi$.
Let $\calS=\langle g_1,..., g_{n+m}\rangle\le \tildecalP_{n+m}$ be the stabilizer group of $\{\mathbb{P}_{\Phi_b}\}_b$.
As discussed in~\cref{sec:clifford_povms}, the Pauli operators reconstructed by the effective POVM are the $\pi_n(g)$ such that $g\equiv \pi_n(g)\otimes\pi_m(g)\in\calS$ and $\langle\psi|\pi_m(g)|\psi\rangle\neq0$.

The entanglement structure of stabilizer states can be characterised, as shown in Ref.~\cite{fattal2004entanglement}, by finding generators for $\calS$ such that
\begin{equation}\label{eq:stab_ent_structure}
\begin{gathered}
    \calS=\bigg\langle
    \{a_i\otimes I\}_{i=1}^{\dim(\calS_n)}
    \cup
    \{I\otimes b_i\}_{i=1}^{\dim(\calS_m)}
    \\
    \cup
    \,\{g_i^{(n)}\otimes g_i^{(m)}\}_{i=1}^p
    \cup \{\bar g_i^{(n)}\otimes \bar g_i^{(m)}\}_{i=1}^p
    \bigg\rangle,
\end{gathered}
\end{equation}
with (i) $\calS_n,\calS_m\le \calS$ the subgroups of operators with support only on $\calH_n$ and $\calH_m$, respectively, (ii) each $g_i^{(n)}$, $g_i^{(m)}$ anticommuting with $\bar g_i^{(n)}$, $\bar g_i^{(m)}$, respectively, and commuting with all other generators, and (iii) $g_i^{(n)},\bar g_i^{(n)}\neq a_j$, $g_i^{(m)},\bar g_i^{(m)}\neq b_j$, for all $i,j$.
The parameter $p$ quantifies the entanglement of the states and satisfies $2p=n+m-\dim(\calS_n)-\dim(\calS_m)$ and $0\le p\le \min(n,m)$.
States are separable for $p=0$ and maximally entangled for $p=\min(n,m)$.
From this decomposition one also finds the relations
\begin{equation}\label{eq:entanglement_relations}
\begin{gathered}
    n=\dim(\calS_n)+p,
    \qquad
    m=\dim(\calS_m)+p,
    \\
    \dim(\{\pi_n(g):\, g\in\calS\})=2p+\dim(\calS_n) = n+p.
\end{gathered}
\end{equation}
Thus the effective POVM contains precisely $2^{n+p}$ linearly independent Pauli operators, provided $\ket\psi$ is such that $\langle\psi|\pi_m(g)|\psi\rangle\neq0$ for sufficiently many $g\in\calS$.
A direct way to conclude that $\dimspanmu\le 2^{n+p}$ is to consider the quotient space $\calS/\calS_m$. By definition of quotient space there is a bijection between the set of projections $\{\pi_n(g):\, g\in\calS\}$ and $\calS/\calS_m$, and $\lvert\calS/\calS_m\rvert=\lvert\calS\rvert/\lvert\calS_m\rvert=2^{n+m}/2^{m-p}=2^{n+p}$.
\end{proof}

Note that there is always some $\ket\psi$ that maximizes $\dimspanmu$, as a random pure state almost surely has nonzero expectation value on all Pauli operators, and thus $\dimspanmu=2^{n+p}$.
In conclusion, an IC measurement is possible in this setting iff $p=n$, which requires $m\ge n$ and maximal entanglement between $\calH_n$ and $\calH_m$.

\begin{corollary}\label{cor:rank_multiples}
    Under the same assumptions as~\cref{thm:entanglement_and_reconstruction}, $\dimspanmu= k\, 2^{n-p} $ for some $k\in\mathbb{N}$.
\end{corollary}
\begin{proof}
    In~\cref{thm:entanglement_and_reconstruction} we used the quotient space $\calS/\calS_m$ to calculate the number of operators in $\calH_n$ reconstructed by the measurement.
    By the same logic, considering instead the quotient space $\calS/\calS_n$, we see that for each $\pi_m(g)\in\tildecalP_m$ there are $\lvert\calS_n\rvert=2^{n-p}$ distinct $\pi_n(g)$ such that $\pi_n(g)\otimes\pi_m(g)\in\calS$.

    Going further, we observe that the subgroups $\calS_n,\calS_m\le\calS$ commute pairwise, $[\calS_n,\calS_m]=0$, and thus their product $\calS_n\calS_m\le\calS$ is also a subgroup, and has order $2^{n+m-2p}$.
    The associated quotient space $\calS/\calS_n\calS_m$ has $2^{n+m}/2^{n+m-2p}=2^{2p}$ elements.
    This quotient space provides a partition of $\calS$ into pairs of cosets of the form $g_i^{(n)}\calS_n\times g_i^{(m)}\calS_m$ and $\bar g_i^{(n)}\calS_n\times \bar g_i^{(m)}\calS_m$, linking all sets of $2^{n-p}$ projections $\pi_n(g)\in\tildecalP_n$ that are paired to the same set of $2^{m-p}$ projections $\pi_m(g)\in\tildecalP_m$.
    Indeed, there is a one-to-one mapping between the cosets in $\calS/\calS_n\calS_m$ and the set of $2p$ nonlocal generators of $\calS$ in~\cref{eq:stab_ent_structure}.

    Thus, if $\ket\psi$ does not annihilate \textit{any} of the $\calH_m$ elements in such a coset, then \textit{all} the corresponding $2^{n-p}$ projections $\pi_n(g)$ are reconstructible by the POVM.
\end{proof}

\hypertarget{ref:sec1}{Explicit example applications of~\cref{thm:entanglement_and_reconstruction,cor:rank_multiples} are given in~\cref{ex:ent1,ex:ent2,ex:ent3}.}

\parTitle{Knowing $\calS_m$ and $p$ is sufficient}
One aspect emerging from the discussion in~\cref{cor:rank_multiples} and~\cref{ex:ent1,ex:ent2,ex:ent3} is that $\calS_m$ and $p$ alone are sufficient to compute $\dimspanmu$.
Indeed, we have in general $C(\pi_m(\calS_m))=\pi_m(\calS)$, and any $\calS$ compatible with a given $\calS_m$ can be obtained taking the quotient space $C(\pi_m(\calS_m))/\pi_m(\calS_m)$ and attaching an $\calH_n$ operator to each resulting coset ensuring the proper commutation properties are respected.
The way this last step is carried out does not affect $\dimspanmu$, but only \textit{which} directions the measurement reconstructs.
Knowledge of $p$ is then necessary to know $\dimspanmu$, as each coset in $\calS/\calS_m$ that is not annihilated by $\ket\psi$ contributes $2^{n-p}$ directions to the overall value of $\dimspanmu$.
In particular, when $m\ge n$ and the entanglement is maximal, we have $p=n$ and $2^{n-p}=1$, thus instead of looking at the ``double-coset'' quotient spaces $\calS/\calS_n\calS_m$ it is sufficient to look at the quotient spaces $\calS/\calS_m$.

\parTitle{The role of initial ancillary states}
Although $p$ controls the maximal value of $\dimspanmu$, the choice of $\ket\psi$ is also important.
As shown in~\cref{sec:clifford_povms}, if $\ket\psi$ is a stabilizer state, then we always have $\dimspanmu=2^n$, independent of $p$.
Stabilizer projections are worst-case scenarios in this sense.
In contrast, a random projection almost surely removes no full coset, thus ensuring $\dimspanmu=2^{n+p}$.
The intermediate scenarios, for example when projecting on structured non-stabilizer states such as $\ket T\equiv T\ket+$, remains a nontrivial open question, and will be analysed below.

\hypertarget{ref:ent_part2}{\Cref{ex:ent4,ex:ent5}} show explicitly how the coset $\calS_m$ determines $\dimspanmu$ for all choices of $\ket\psi$ via the quotient $C(\pi_m(\calS_m))/\pi_m(\calS_m)$.

\parTitle{Maximal \texorpdfstring{$\dimspanmu$}{dim(mu)} and entanglement}
When considering the $\dimspanmu$ resulting from specific choices of $\ket\psi$, the double coset decomposition in~\cref{cor:rank_multiples} might lead one to wonder whether higher $\dimspanmu$ might occur at lower entanglement values, which seems counterintuitive.
Alas, this \textit{is} possible, and intuitively can be traced back to the fact that the information in entangled states is intrinsically nonlocal, and can be destroyed by local operations.
As a simple example, consider two-qubits with $n=m=1$, and $\calS=\langle\ps{ZZ},\ps{XX}\rangle$, \textit{i.e.} a Bell measurement.
If the second qubit's ancilla is maximally mixed, $\rho_R=I/2$, the effective measurement becomes trivial:
\begin{equation}
    \mu_b = \Trace_R[\mathbb{P}_{\Phi_b}(I\otimes I/2)] = \frac I2,
    \,\, \forall b.
\end{equation}
for maximally entangled $\ket{\Phi_b}$.
Thus, in this case, a non-entangled measurement such as a computational basis one achieves $\dimspanmu=2^n$, which is higher than the $\dimspanmu=1$ obtained with a fully entangled measurement.
This phenomenon does not seem to occur when $\ket\psi$ is pure, due to the purity constraint restricting the possible patterns of nonzero Pauli expectation values.
In lack of a formal proof, we leave this statement as a conjecture:
\begin{conjecture}\label{conj:maxentmaxrank}
    For any pure $\ket\psi$, the maximal value of $\dimspanmu$ can be obtained with maximally entangled $\calS$.
\end{conjecture}
\textit{Discussion. }
    This is clearly true when $\ket\psi$ is stabilizer, or --- with probability $1$ --- for random $\ket\psi$.
    More generally, we showed that  $\dimspanmu= 2^{n+p}$ is achievable with suitable choices of $\ket\psi$, and that for any $\ket\psi$ we have $\dimspanmu=k\, 2^{n-p}$ for some integer $2^p\le k\le 2^{2p}$.
    For a given $\ket\psi$ to give $\dimspanmu=k\, 2^{n-p}$, it must satisfy $\langle\psi|P|\psi\rangle\neq0$ for at least $k$ Pauli operators distributed in $k$ distinct cosets of $\calS/\calS_n\calS_m$.
    Thus, although for smaller $p$ the maximal $\dimspanmu$ is smaller, $\ket\psi$ needs to satisfy fewer conditions to achieve it, and vice versa.
    
    Proving this statement thus requires to show that given an $\calS$ with entanglement $p$, and $\ket\psi$ corresponding to $\dimspanmu= k\, 2^{n-p}$, there is some $\calS'$ with entanglement $p+1$ such that \textit{the same} $\ket\psi$ achieves $\dimspanmu=k'\, 2^{n-p-1}$ with $k'\ge 2k$.

    In fact, we can state this even more explicitly noting that all these statements about $\dimspanmu$ are invariant under a unitary change of basis local to $\calH_n$ and $\calH_m$.
    Furthermore, there is always a basis change via a Clifford unitary local to $\calH_n$ and $\calH_m$ such that in the new basis $\calS$ has local generators $a_1=\ps{Z}_1,\dots, a_{n-p}=\ps{Z}_{n-p}$, $b_1=\ps{Z}_{1}, \dots, b_{m-p}=\ps{Z}_{m-p}$, and nonlocal generators $g_i^{(n)}\otimes g_i^{(m)}=\ps{X}_{n-p+i}\ps{X}_{m-p+i}$, $\bar g_i^{(n)}\otimes \bar g_i^{(m)}=\ps{Z}_{n-p+i}\ps{Z}_{m-p+i}$, $i=1,\dots, p$.
    Intuitively, this measurement basis involves Bell measurements on $p$ qubit pairs, and local $Z$ measurements on the remaining ones.
    With this choice of $\calS$, the $\mu_b'$ measurement operators take the form (modulo some suitable relabeling of outcomes):
    \begin{equation}\label{eq:mubprime_canonical_S}
        \prod_{i=1}^{n-p} Z_i^{a_i}
        \prod_{j=1}^p P_{n-p+j}^{(j)}
        \langle\psi|
        \left(\prod_{k=1}^{m-p} Z_k^{b_k}\prod_{\ell=1}^p P_{m-p+\ell}^{(\ell)}\right)
        |\psi\rangle,
    \end{equation}
    where $a_i,b_k\in\{0,1\}$, the lower indices indicate on which qubit the operator is acting, and each $P^{(j)}\in\{I,X,Y,Z\}$ for each $j$.
    Each choice of $(a_i)_{i=1}^{n-p}$, $(b_k)_{k=1}^{m-p}$, and $(P^{(j)})_{j=1}^p$, corresponds to a different operator $\mu_b'$.
    Note in particular that for all $j=1,...,p$, the operator $P^{(j)}$ outside the bracket and the one inside the bracket are equal (though applied to different qubits).
    The value of $\dimspanmu$ equals the number of nonzero terms in~\cref{eq:mubprime_canonical_S}.
    In turn, this is equal to the number $r$ of operators $\prod_{k=1}^{m-p} Z_k^{b_k}\prod_{\ell=1}^p P_{m-p+\ell}^{(\ell)}$ with nonzero expectation value on $\ket\psi$, multiplied by $2^{n-p}$.
    One way to prove the conjecture is thus to show that if $\ket\psi$ gives some value of $\dimspanmu$, then increasing $p$ by $1$, we get $\dimspanmu'\ge 2\dimspanmu$.
\hypertarget{ref:ent_part3}{In~\cref{ex:ent6} we prove the statement when $n=m=2$. The general case is left as a conjecture.}

\section{Reconstruction with \texorpdfstring{$t$}{t}-doped POVMs}
\label{sec:tdoped_POVMs}

We now broaden the scope to effective measurements obtained with \textit{non‑stabilizer} circuits.
Here, we extend the results of~\cref{sec:clifford_povms} to POVMs~\eqref{eq:effective_POVM} where $U$ is $t$-doped, $\rho_R=\mathbb{P}_\psi$ is a stabilizer state, and the physical measurement is stabilizer.
We refer to such measurements, schematically represented in~\cref{fig:parallel_doping}, as \textit{$t$-doped POVMs}.
The number of $T$ gates injected in a Clifford circuit is a standard way to quantify non-stabilizerness~\cite{bravyi2005universal,leone2021renyi,jiang2023Lower,leone2024Stabilizer,ahmadi2024Mutual,haug2025Probing}.
We will argue that an $n$-qubit IC-POVM requires a circuit with at least $2n$ $T$ gates.
More specifically, we prove that $t\ge \frac{2n}{\log_23}\approx 1.26 n$ is necessary, and provide evidence that the actual threshold is $t\ge 2n$. We also give explicit constructions proving that $t=2n$ is sufficient for all $n$.
Interestingly, these bounds match the quasi-chaotic and chaotic bounds for $t$-doped circuits in the sense of~\cite{leone2024learning}.

\begin{figure}[tb]
    \centering
    \includegraphics[width=0.5\linewidth]{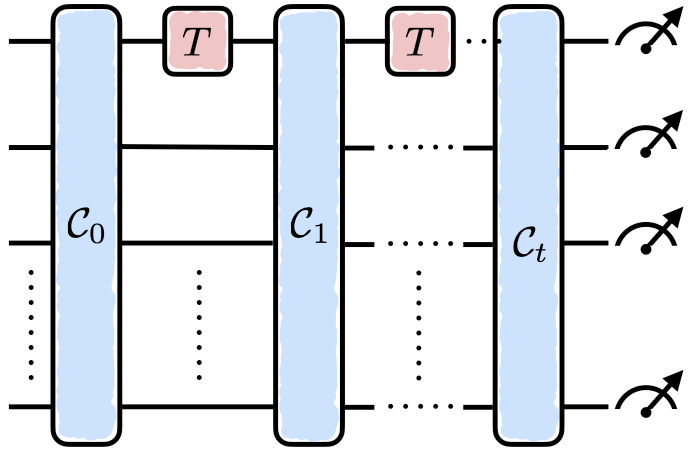}
    \hspace{4mm}
    \includegraphics[width=0.35\linewidth]{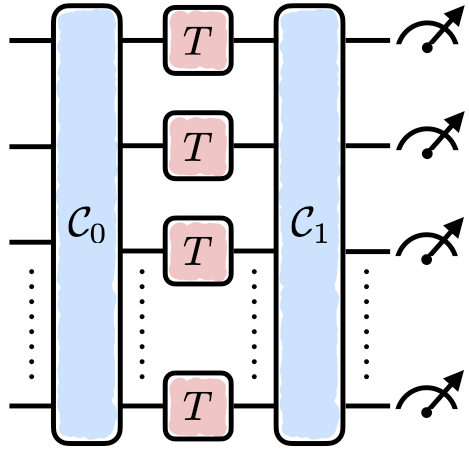}
    \caption{Serial (left) and parallel (right) doping of an $n$-qubit Clifford circuit.
    Each layer $\mathcal{C}_0$ is a random Clifford gate.
    The final measurement is performed in the computational basis.
    }
    \label{fig:parallel_doping}
\end{figure}

\begin{figure}[tb]
    \centering
    \includegraphics[width=0.8\linewidth]{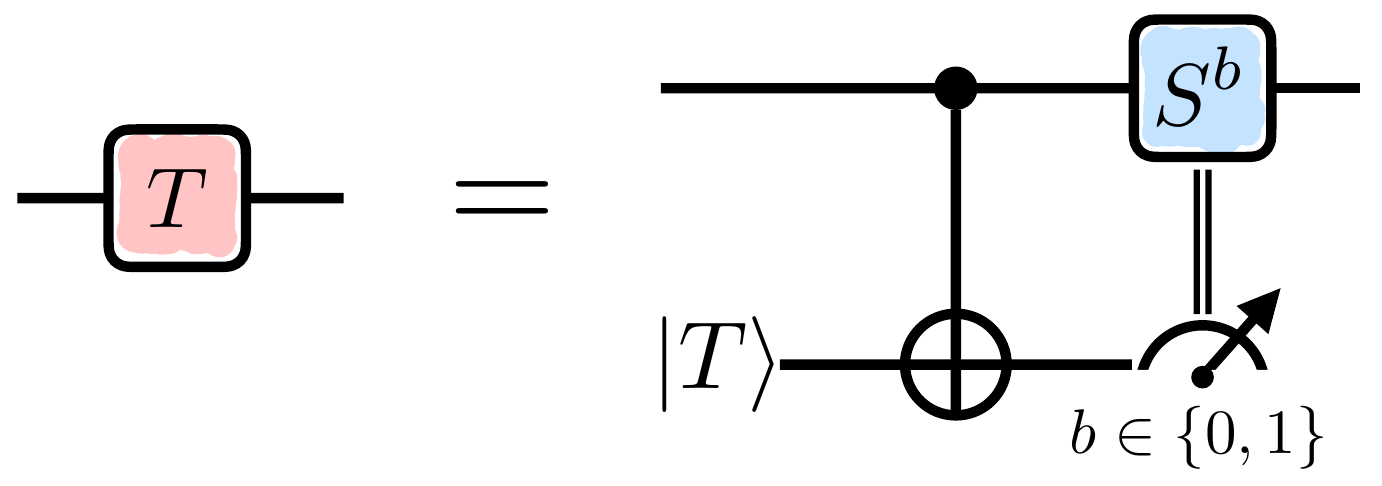}
    \caption{Gadget to implement a $T$ gate adding an ancilla initially in $\ket T\equiv T\ket +$.
    Measuring the second qubit after a CNOT, teleports a $T$ gate on the first qubit when outcome is $\ket 0$, and a $ST$ gate otherwise.
    }
    \label{fig:gadget}
\end{figure}

\parTitle{$T$ gate gadgets}
Circuits containing $T$ gates leave the Clifford group, so the stabilizer formalism no longer applies directly. 
A useful formal workaround is to replace each $T$ gate with a \textit{T-gadget}~\cite{zhou2000methodology,bravyi2016improved}. Each gadget, shown schematically in~\cref{fig:gadget}, introduces an ancilla qubit prepared in the magic state $\ket T\equiv T\ket+=\frac{\ket0+e^{i\pi/4}\ket1}{\sqrt2}$, applies a CNOT with the ancilla as target, and finally projects the ancilla onto $\ket0$.
In general, gadgets involve a computational basis measurement on the ancilla rather than a projection, as shown in~\cref{fig:gadget}. However, since we use them purely as a mathematical device, which never needs to be experimentally implemented for our purposes, we may restrict to the $\ket 0$ measurement outcome and thus equivalently treat it as a projection.
This simplification also removes the need for the corrective $S$ gate and makes the circuit more compact.
This sub-circuit reproduces exactly the action of the original $T$ gate.
After substituting all $T$ gates with their gadgets, a circuit with $t$ such gates becomes the Clifford circuit $\tilde U=C_0 \prod_{k=1}^t \on{CNOT}_k\, C_k$, where each $C_k$ is a Clifford subcircuits between consecutive $T$ gates, and $\on{CNOT}_k$ are the gadgets' CNOTs.
The final physical measurement then consists of a stabilizer measurement on $n+m$ qubits together with a projection of the $t$ gadget ancillas onto $\ket0$, as shown in~\cref{fig:explicit_gadget_example}.

The non‑stabilizerness is thus isolated in the preparation of the magic states, and in the Heisenberg picture we can now describe the evolution of the POVM remaining within the stabilizer formalism.
The effective POVM on the initial $n+m+t$ qubits, \textit{before} projecting onto $\ket\psi$, is described by a $(n+m+t)$-qubit stabilizer group with $t$ fixed syndromes, and has exactly $2^{n+m}$ outcomes.
Explicitly, this means that the Heisenberg-evolved measurement states $\ket{\Phi_b}$ form the common eigenbasis of the $n+m+t$ generators of $\calS$, with $n+m$ eigenvalues free to take any value in $\{0,1\}^{n+m}$, but with the eigenvalues corresponding to the last $t$ generators fixed to $+1$.
An explicit example of a $t$-doped circuit in this representation is shown in~\cref{fig:explicit_gadget_example}, and a detailed derivation of the corresponding effective measurement and stabilizer structure is provided in~\cref{ex:tdoped_2}.

\begin{figure}[tbp]
    \centering
    \includegraphics[width=\linewidth]{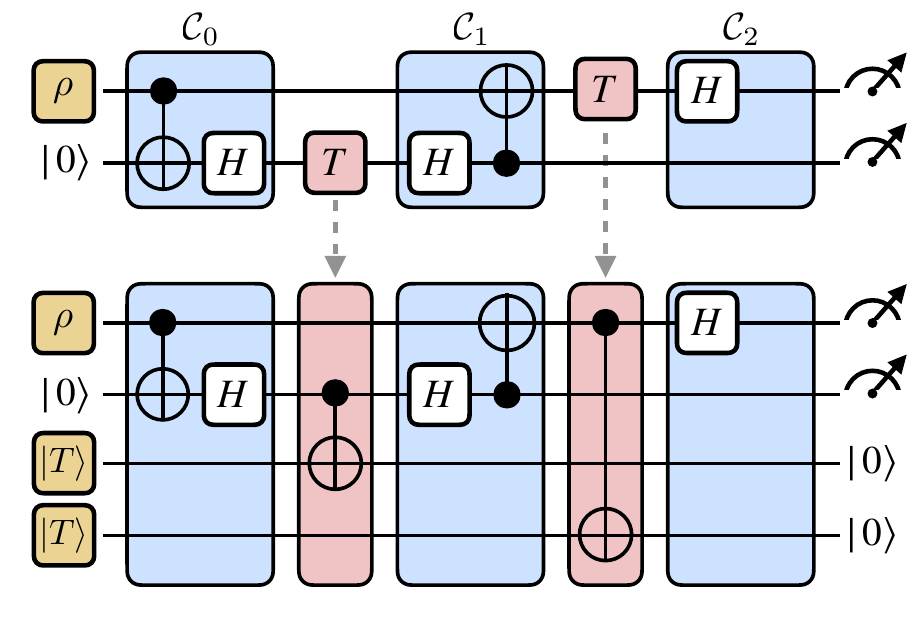}
    \caption{
    Example of $t$-doped circuit with $n=m=1$ and $t=2$.
    Using the gadgets trick, the $T$ gates become CNOTs with target a new ``gadget ancilla'', whose initial value is $\ket T$ and that is projected onto $\ket0$ at the end of the circuit.
    The explicit characterisation of this circuit in the stabilizer formalism is given in~\cref{ex:tdoped_2}.
    }
    \label{fig:explicit_gadget_example}
\end{figure}

\parTitle{$T$ gates and ancillas}
A pressing question is: given $n$ data qubits, $m$ stabilizer ancillas, and a unitary with $t$ $T$ gates, what values of $\dimspanmu$ are possible?
We already know the answer in some simple cases:
\begin{itemize}
    \item With no ancillas ($m=0$), $\dimspanmu=2^n$ regardless of $t$ (\cref{lemma1}).
    \item With no $T$ gates ($t=0$), $\dimspanmu=2^n$ regardless of $m$ (\cref{thm:structural_clifford_povms}).
    \item For large $t$, the unitary can be effectively arbitrary, so informational completeness is possible iff $m\ge n$.
\end{itemize}

To handle the general case, we utilise the gadget picture.
As discussed before, we can characterize the full circuit by a stabilizer group $\calS$ on $n+m+t$ qubits, with $t$ fixed syndromes corresponding to the final projections on the gadget qubits.
Due to \cref{thm:structural_clifford_povms_plus}, projecting the $m$ ancillas removes $m$ generators.
However, because $t$ generators already have fixed syndromes, it matters whether the ancilla projection removes generators that are already fixed or generators that are associated to both eigenvalues $\pm1$. More specifically:
\begin{itemize}
    \item If $m\ge t$, in the best-case scenario the projection removes the $t$ frozen plus $m-t$ active generators, leaving $(n+m)-(m-t)=n+t$ free generators.
    In the worst-case scenario, the projection removes instead $t$ active generators, leaving only $(n+m)-m=n$ free.
    \item If $m \le t$, in the best-case scenario the projection removes only frozen generators, thus keeping all the $n+m$ free ones.
    In the worst-case scenario, it instead removes $m$ of the active generators, leaving again with $(n+m)-m=n$ free ones.
\end{itemize}
We conclude that the post-projection stabilizer group has between $n$ and $n+\min(t,m)$ free generators.
We will almost always focus here on the best-case-scenario cases, where after projecting the ancillas we still have $n+\min(t,m)$ free generators, and furthermore consider the situation when there are enough ancillas, $m\ge t$. In these cases we can simply assume that the stabilizer measurement basis before the ancilla projection is described by an $(n+t)$-qubit stabilizer basis.
\hypertarget{ref:tdoped_part1}{See~\cref{ex:tdoped_1,ex:tdoped_2} for explicit applications of this formalism, and how Heisenberg-evolved measurement states can be described via ``stabilizer groups with some fixed syndromes''.}

\begin{theorem}\label{thm:easy_tn_bound}
    IC $t$-doped POVMs are possible only if $t\ge \frac{2}{\log_2 3}n\approx 1.26 n$.
\end{theorem}
\begin{proof}
    In the gadget picture, assume the ancilla projection yields an $(n+t)$-qubit stabilizer measurement.
    IC requires the existence of \textit{at least} $4^n$ strings $g\in \calS$ such that $\langle\psi|\pi_t(g)|\psi\rangle\neq0$.
    The states $\ket T\equiv T\ket+=\frac1{\sqrt2}(\ket0+e^{i\pi/4}\ket1)$ satisfy
    \begin{equation}
        \langle X,\mathbb{P}_T\rangle=\langle Y,\mathbb{P}_T\rangle=\frac{1}{\sqrt2}\,,\quad\langle Z,\mathbb{P}_T\rangle\,=\,0.
    \end{equation}
    Thus we want the surviving $g\in\calS$ are all and only those such that $\pi_t(g)$ contains no $\ps{Z}$ operators; we shall refer to such strings as $\ps{Z}$\textit{-free} strings.
    The total number of possible $\ps{Z}$-free $t$-qubit strings is $3^t$.
    Therefore to achieve informational completeness there must be $4^n$ strings $g\in\calS$ each one associated to a $\ps{Z}$-free substring $\pi_t(g)\in\tildecalP_t$. This gives the necessary condition:
    \begin{equation}
        3^t\ge 4^n
        \iff t\ge \frac{2}{\log_23}n\approx 1.26\, n.
    \end{equation}
\end{proof}
\begin{proof}[Alternative proof]
    This is a different proof of the same statement. We include it because it uses rather different ideas and could therefore be useful as reference for extensions of the result.
    
    Assume as before that projecting the ancillas we are left with an effective stabilizer measurement on the remaining $n+t$ qubits.
    We know from~\cref{sec:entanglement_and_reconstruction} that maximal entanglement, $p=n$, is a necessary condition for full reconstruction, so we also assume it here.
    We thus have by construction $\lvert\calS\rvert=2^{n+t}$, and $\lvert\calS_n\rvert=1$, $\lvert\calS_t\rvert=2^{t-n}$, where $\calS_n,\calS_t$ are the subgroups of $\calS$ that act nontrivially only on $\calH_n$ and $\calH_t$, respectively.
    Furthermore, $\pi_t$, which project Pauli strings on their $\calH_t$ component, is injective on $\calS$, because there is no $g\in \calS$ such that $\pi_t(g)=I$ --- again by the structure of maximally entangled stabilizer states discussed in~\cref{sec:entanglement_and_reconstruction}.
    Finally, we remember that the number of cosets in the quotient group $\calS/\calS_t$ is precisely $4^n$, as
    $\lvert\calS/\calS_t\rvert=
    \lvert\calS\rvert/\lvert\calS_t\rvert=2^{n+t}/2^{t-n}=4^n$.

    We are interested in the cases where each coset $[g]\in\calS/\calS_t$ contains at least one element $g\in[g]$ such that $\pi_t(g)$ is $Z$-free.
    Denote with $\calZ_f$ the set of $Z$-free Pauli operators $g\in \calP_t$. Its size is $\lvert \calZ_f\rvert=3^t$.
    We can then reformulate our requirement as the following constraint:
    \begin{equation}\label{eq:IC_condition}
        (\calZ_f\cap\pi_t(\calS))\pi_t(\calS_t)=\pi_t(\calS). 
    \end{equation}
    The LHS represents here the product of the groups $(\calZ_f\cap\pi_t(\calS))$ and $\pi_t(\calS_t)$, which is the set whose elements are all possible products of elements taken from the two individual groups.
    To understand this condition, observe that $\calZ_f\cap\pi_t(\calS)$ is the set of all $\calH_t$-projections of elements $g\in\calS$ that are $Z$-free, and $\calS_t$ is the generator of the cosets, thus the product $(\calZ_f\cap\pi_t(\calS))\pi_t(\calS_t)$ contains all elements in the cosets, projected on the $t$ qubits, that can be obtained building cosets from the $Z$-free elements.
    Thus~\cref{eq:IC_condition} can be read as stating that we can generate all cosets using only $Z$-free elements as representatives.

    For any abelian group $G$, finite set $X\subseteq G$, and subgroup $H\le G$, we have $\lvert XH\rvert\le \lvert X\rvert \cdot\lvert H\rvert$.
    This tells us that
    \begin{equation}
    \begin{gathered}
        \lvert
        (\calZ_f\cap\pi_t(\calS))\pi_t(\calS_t)
        \rvert
        \le 
        \lvert \calZ_f\cap\pi_t(\calS))\rvert
        \cdot
        \lvert \pi_t(\calS_t)\rvert
        \\ \le \lvert \calZ_f\rvert\cdot\lvert \pi_t(\calS_t)\rvert = 3^t 2^{t-n}.
    \end{gathered}
    \end{equation}
    Using this in~\cref{eq:IC_condition} gives us the bound
    \begin{equation}
        2^{n+t}=\lvert \pi_t(\calS)\rvert
        \le 3^t 2^{t-n}
        \iff 4^n \le 3^t
    \end{equation}
\end{proof}

The bound in~\cref{thm:easy_tn_bound} is a necessary but far from sufficient condition.
Even when $3^t\ge 4^n$, there is no guarantee that each coset actually contains a $\ps{Z}$-free element.
We know in particular that $2n$ $T$ gates are sufficient:
\begin{theorem}\label{thm:tgen}
    IC $t$-doped POVMs are possible for all $t\ge2n$.
\end{theorem}
\begin{proof}
    Set $t=2n$ and assume $p=n$.
    We want to find a $(n+t)$-qubit stabilizer group $\calS$ such that each coset in the quotient $\calS/\calS_t$ contains at least one element that is $\ps{Z}$-free in its $\calH_t$ component.
    Equivalently, it suffices to find a $t$-qubit abelian group $\pi_t(\calS_t)$ such that each coset in $C(\pi_t(\calS_t))/\pi_t(\calS_t)$ has a $\ps{Z}$-free representative.

    Take $\calS_t\equiv \langle h_1,..., h_{t-n}\rangle$ with $h_i=\ps{I}_n\ps{X}_{2i-1}\ps{Z}_{2i}$, that is, take generators that have disjoint support, act trivially on $\calH_n$, and like $\ps{XZ}$ on their support.
    We already know from~\cref{ex:ent4} that $\ps{XZ}$ induces $\ps{Z}$-free cosets.
    More specifically, a single such generator induces an embedding of $\tildecalP_1$ into two-qubit cosets with one $\ps{Z}$-free representative each:
    \begin{equation}
        \ps X\to \{\ps{XI},\ps{IZ}\},
        \quad
        \ps Y\to \{\ps{YY},\ps{ZX}\},
        \quad
        \ps Z \to \{\ps{YZ},\ps{ZI}\}.
    \end{equation}
    The same scheme works for $t=2n$ qubits, generalised as
    \begin{equation}
    \begin{gathered}
        \ps X_i\to \{\ps{X}_{2i-1},\ps{Z}_{2i}\},
        \quad
        \ps Y_i\to \{\ps{Y}_{2i-1}\ps{Y}_{2i},\ps{Z}_{2i-1}\ps{X}_{2i}\},
        \\
        \ps Z_i \to \{\ps{Y}_{2i-1}\ps{Z}_{2i},\ps{Z}_{2i-1}\}.
    \end{gathered}
    \end{equation}
    An explicit $t$-doped circuit with $t=2n$ that is IC is reported in~\cref{app:circ}.
\end{proof}

In fact, we strongly believe that IC $t$-doped POVMs require $t\ge 2n$, although we do not have a formal proof of this stronger statement, which we therefore leave as a conjecture:

\begin{conjecture}\label{conj:2n}
    There are IC $t$-doped POVMs iff $t\ge 2n$.
\end{conjecture}
\textit{Discussion.}
We already showed that IC circuits exist for $t\ge 2n$. Exploiting the fact that the informational-completeness of any given circuit can be fully characterised by the $\calS_m$ subgroup, as discussed in~\cref{sec:entanglement_and_reconstruction}, we performed an exhaustive numerical search and verified that no IC $t$-doped POVM with $t<2n$ exist for $n=1,2,3$. We also sampled $10^6$ random $t$-doped circuits with $n=m=4$, various $t$, and never found examples of IC circuits for $t<2n$. The $t=2n$ threshold is also perfectly consistent with the bounds in~\cref{thm:rank_tlen,thm:rank_tgen} proved below.

The situation somewhat simplifies when $t\le n$:
\begin{theorem}\label{thm:rank_tlen}
    $t$-doped POVMs corresponding to maximally entangled stabilizer states, for $t\le n,m$, have $\dimspanmu\le 2^n(\frac32)^t$.
\end{theorem}
\begin{proof}
    We again operate under the assumption that projecting the ancillas leaves behind an $(n+t)$-qubit stabilizer group with $n+t$ free generators.

    Maximal entanglement with $t\le n$ ensures that $\calS_t$ is trivial, and all $4^t$ strings appear in $\{\pi_t(g):\, g\in\calS\}$, each one corresponding to one of the $2^{n-t}$ cosets in $\calS/\calS_n$.
    Of the $4^t$ strings in $\tildecalP_t$, precisely $3^t$ are $\ps{Z}$-free.
    Thus the corresponding $\dimspanmu$, which equals the number of $\pi_n(g)$ corresponding to a $\ps{Z}$-free $\pi_t(g)$, is at most $2^{n-t} 3^t=2^n(\frac32)^t$.
\end{proof}
Note the consistency of the statement of~\cref{thm:rank_tlen} with~\cref{cor:rank_multiples}: we work here in the maximally entangled case with $t\le n$, thus $p=t$. Thus~\cref{thm:rank_tlen} tells us that the maximum of $\dimspanmu$ must be a multiple of $2^{n-p}=2^{n-t}$, which is precisely what we also found here.
\hypertarget{ref:tdoped_part2}{See~\cref{ex:tdoped_3}} for an application example of this formalism.

We thus found that all maximally entangled cases with $t \le n$ --- provided information survives the ancilla projection in the gadget picture --- give $\dimspanmu=2^n (3/2)^t$. This is because in these cases all $4^t$ strings appear projecting on $\tildecalP_t$, and each one corresponds to a fixed number of strings on $\tildecalP_n$.
This contrasts with what happens for $t>n$, where all $\pi_n(g)\in\tildecalP_n$ appear in the coset decomposition $\calS/\calS_t$, but the associated cosets contain multiple elements that might or might not survive the projection onto $\ket{T}^{\otimes t}$.
Indeed, for $t>n$, having maximal entanglement, $p=n$, ensures that \textit{for some choice of $\ket\psi$} the measurement is IC, but fixing $\ket\psi=\ket T^{\otimes t}$ complicates things considerably, as it is often the case that entire cosets are annihilated by the projection, thus reducing $\dimspanmu$. Nonetheless, we have the following:
\begin{theorem}\label{thm:rank_tgen}
    For $t> n$, we have $\dimspanmu\le A$ for some
    \begin{equation}\label{eq:formula_weird_rank}
        A \ge 2^{-\ell} (3^{a+1}-1)^r (3^a-1)^{\ell-r},
    \end{equation}
    with $a\equiv \lfloor \frac{t}{\ell}\rfloor$, $r\equiv t-a\ell$, and $\ell\equiv t-n$.
\end{theorem}
\begin{proof}
    Here $A$ represents the true maximum achievable value of $\dimspanmu$, when maximising it over all possible $t$-doped circuits with fixed $n,m$.
    This theorem is phrased this way because we do not prove what the real maximum value of $\dimspanmu$ is, but rather just find examples of measurements where $\dimspanmu=A$.
    Although we did not find any such example by numerical search, it is possible that there are other measurements that give $\dimspanmu>A$.

    We prove the statement providing explicit constructions in terms of generators for $\calS/\calS_t$, fixing $p=n$ so that $\calS_n=\{\ps{I}\}$ and $\calS_t$ has $\ell=t-n$ generators.

    For $t=n+1$, $\ell=1$, take a generator with all $X$ and a single $Z$, such as $\pi_t(\calS_t)=\langle \ps{X\cdots XZ} \rangle$.
    Then each coset contains at most one $\ps{Z}$-free string: indeed, for a given $g$ to commute with $\ps{X}\cdots\ps{XZ}$, either $g_{n+1}$ (the $(n+1)$-th qubit in $g$) commutes with $\ps{Z}$, hence $g_{n+1}=\ps{I}$ and thus $(g\cdot \ps{X\cdots XZ})_{n+1}=\ps{Z}$, or some other element of $g$ must anticommute with $\ps{X}$, thus being equal to $\ps{Y}$, and again multiplying by $\ps{X\cdots XZ}$ we would get a $\ps{Z}$ in the resulting string.
    Furthermore, the total number of $\ps{Z}$-free elements in the centraliser $C(\ps{X\cdots XZ})$ is $(3^t-1)/2$, as shown in~\cref{sec:Zfree_centralisers_count}.
    Thus $(3^t-1)/2$ is precisely the number of cosets with $\ps{Z}$-free elements.
    This matches~\cref{eq:formula_weird_rank} because $a=t$, $r=0$.

    For $t=n+2$, $\ell=2$, $\dimspanmu=A$ is achieved with pairs of generators with (i) disjoint support, (ii) each one having the $\ps{X\cdots XZ}$ pattern of the $\ell=1$ case, and (iii) with supports divided among the $t$ qubits as evenly as possible.
    So for example for $n=2$, this means to take $\pi_t(\calS_t)=\langle \ps{XZII},\ps{IIXZ}\rangle$, for $n=3$ to $\pi_t(\calS_t)=\langle \ps{XXZII}, \ps{IIIXZ}\rangle$, etc.
    By an argument similar to the one used for $\ell=1$, each coset generated by these $\pi_t(\calS_t)$ contains at most one $\ps{Z}$-free element, and furthermore its centraliser is the product of the centralisers of each of its generators.
    Thus $\lvert C(\pi_t(\calS_t))\rvert$ equals $\frac1{2^2}(3^{t/2}-1)^2$ for even $t$, and $\frac1{2^2}(3^{(t+1)/2}-1)(3^{(t-1)/2}-1)$ for odd $t$, which can be written concisely as~\cref{eq:formula_weird_rank}.

    The same pattern continues for larger $\ell$. In each case, we take $\ell$ generators with disjoint support, each containing the $\ps{X\cdots XZ}$ pattern, and dividing the available $t$ by $\ell$ as evenly as possible.
    Then multiplying the centralisers of these generators via~\cref{sec:Zfree_centralisers_count} we get the result.
\end{proof}
Note in particular that for $t=2n$, we have $\ell=n$, $a=2$, $r=0$, and we recover the IC case because $2^{-n}(3^2-1)^n=4^n$.
While for $t=n$ we have $\dimspanmu=3^n$ from~\cref{thm:rank_tlen}, and for $t=0$ we have $\dimspanmu=2^n$ because we revert back to a simple stabilizer measurement on the $n$ data qubits.

\section{Conclusions}
\label{sec:conclusions}

We investigated the information retrievable from single-setting measurement scenarios employing stabilizer operations, both with and without $T$ gates injected in the circuit. We contextually analysed the role of entanglement in determining the amount of information accessible from such measurements.

To our knowledge, this work presents the first investigation of non-stabilizerness from a metrological perspective, and paves the way for the development of single-setting estimation strategies on circuit-based platforms.

Our results bridge two previously distinct research directions --- the resource theory of non-stabilizerness and the theory of quantum measurements --- and show that the very same resource enabling quantum computational advantage is also essential to realize highly informative measurements, an irremissible ingredient in modern estimation protocols such as shadow tomography and quantum machine learning algorithms based on fixed reservoirs~\cite{innocenti2023potential,vetrano2025state}.
Among the potential repercussions of our analysis, we mention the possibility to realize efficient single-setting state-estimation and shadow tomography protocols. Indeed, although such protocols are known to achieve favorable scalings under random measurement ensembles, recent work~\cite{acharya2021ShadowTomographyBased,nguyen2022OptimizingShadowTomography,innocenti2023shadow} has shown that comparable efficiencies are achievable in single-setting scenarios. Our results provide the explicit non-stabilizerness thresholds required for these protocols, thus paving the way for their practical implementation on near-term circuit-based platforms.

Our study suggests several avenues for future research:
\begin{itemize}
    \item \textit{Design and robustness of circuit architectures}.
    A natural next step is to identify what differentiates circuits with identical doping levels but differing reconstruction performance, and to determine how these differences manifest under realistic experimental noise.
    Although all IC circuits are equivalent in the context of this work, informational completeness is necessary but alone does not guarantee good estimation performances.
    A thorough investigation of the resource cost associated with performing, for instance, shadow tomography using different IC $t$-doped POVMs, could inform the design of minimal-length universal circuits for quantum extreme learning machines (QELMs) and, more generally, for single-setting state-tomography protocols that remain robust to finite-statistics effects.
    \item \textit{Interplay between magic and entanglement}. Clarifying the interplay between magic and entanglement remains an important open problem: while the latter is crucial for distributing information and can be generated by Clifford operations, high-magic states often require only local operations, and the two resources likely occupy largely disjoint regions of Hilbert space, with limited overlap~\cite{gu2024doped,gu2025magic}.
    \item \textit{Formal proof of the conjectures}. Formally proving~\cref{conj:2n,conj:maxentmaxrank} would be valuable from a foundational standpoint and would contribute to a clearer understanding of the metrological role of non-stabilizerness in single-setting measurement scenarios.
    \item \textit{Non-stabilizerness cost of QELMs}.
    \Cref{thm:structural_clifford_povms} implies that QELMs, and more generally any single-setting measurement strategy, fundamentally require magic to reconstruct arbitrary observables. This raises the question of the non-stabilizerness cost required by QELMs to enable full reconstruction, when the reservoir interaction is implemented as a $t$-doped circuit.
    Studying the interplay between non-stabilizerness and quantum machine learning is therefore another interesting direction of future research.
\end{itemize}

\acknowledgments
GLM acknowledge funding from the European Union - NextGenerationEU through the Italian Ministry of University and Research under PNRR-M4C2-I1.3 Project PE-00000019 "HEAL ITALIA" (CUP B73C22001250006 ). 
SL, AF, and GMP  
acknowledge support by MUR under PRIN Project No. 2022FEXLYB.
Quantum Reservoir Computing (QuReCo). LI, SL, MP, and GMP acknowldge funding from the “National Centre for HPC, Big Data and Quantum Computing (HPC)” Project CN00000013 HyQELM – SPOKE 10.
MP is grateful to the Royal Society Wolfson Fellowship (RSWF/R3/183013), the Department for the Economy of Northern Ireland under the US-Ireland R\&D Partnership Programme, the PNRR PE Italian National Quantum Science and Technology Institute (PE0000023), and the EU Horizon Europe EIC Pathfinder project QuCoM (GA no.~10032223).

\bibliography{bibliography}

\appendix

\section{Stabilizer formalism}
\label{app:stab}

The stabilizer formalism is widely used in the context of quantum error-correcting codes and fault-tolerant quantum computation techniques. The formalism is built around the properties of the Pauli and Clifford groups. Quantum circuits that involve preparation of a measurement on a computational basis and Clifford gates are known to be classically simulable \cite{aaronson2004improved}.

The \emph{$n$-qubit Pauli group} $\mathcal{P}_n$ is the group generated by the $n$-fold tensor products of the single-qubit Pauli matrices, $\{I,X,Y,Z\}$, along with multiplicative factors of $\pm 1$, $\pm i$. The \emph{Clifford group} $\mathcal{C}_n$ is defined as the normalizer of the Pauli group in the unitary group $U(2^n)$
\begin{equation}
\mathcal{C}_n = \{ U \in U(2^n) \mid UPU^\dagger \in \mathcal{P}_n, \; \forall P \in \mathcal{P}_n \}.
\end{equation}
That is, the Clifford group consists of all unitaries that map Pauli operators to Pauli operators under conjugation. The Clifford group is generated by the Hadamard gate $H$, the phase gate $S$, and the controlled-NOT gate $\text{CNOT}$
\[
H = \frac{1}{\sqrt{2}}\begin{pmatrix} 1 & 1 \\ 1 & -1 \end{pmatrix}, \quad
S = \begin{pmatrix} 1 & 0 \\ 0 & i \end{pmatrix}, \quad
\text{CNOT} = \begin{pmatrix}
I & 0 \\
0 & X 
\end{pmatrix}.
\]
A quantum state $\ket{\psi}$ is called a \emph{stabilizer state} if there exists an abelian subgroup $\mathcal{S} \subset \mathcal{P}_n$, called the \emph{stabilizer group}, such that
\[
P \ket{\psi} = \ket{\psi}, \quad \forall P \in \mathcal{S},
\]
and $\mathcal{S}$ is maximal, i.e., it has $2^n$ elements and stabilizes a unique $n$-qubit state. Each generator of the stabilizer group is a Pauli operator, and there are $n$ independent generators. Stabilizer states include computational basis states, Bell states, GHZ states, and many other entangled states that can be prepared using only Clifford circuits.

The tableau representation provides an efficient and compact way to describe stabilizer states and simulate their evolution under Clifford operations using classical computation. It captures the action of the stabilizer group generators using binary arithmetic over $\mathbb{F}_2$ (the finite field with two elements), enabling simulations that scale polynomially with the number of qubits.
An $n$-qubit stabilizer state is fully described by an abelian group $\mathcal{S}$ of $2^n$ Pauli operators with $n$ independent generators $\{g_1, \dots, g_n\}$. Each generator $g_i$ can be expressed in terms of its tensor product of Pauli operators:
\[
g_i = i^{k_i} X^{\mathbf{x}_i} Z^{\mathbf{z}_i},
\]
where $\mathbf{x}_i, \mathbf{z}_i \in \mathbb{F}_2^n$ are binary vectors indicating the presence of $X$ and $Z$ operators on each qubit, and $i^{k_i}$ is an overall phase factor.

These generators are organized into a binary matrix called the \emph{tableau}, consisting of $n$ rows (one per generator) and $2n + 1$ columns:
\[
\text{Tableau} =
\left[
\begin{array}{c|c|c}
\mathbf{X} & \mathbf{Z} & \mathbf{r}
\end{array}
\right] \in \mathbb{F}_2^{n \times (2n + 1)}.
\]
The left half encodes the $X$ components, the center encodes the $Z$ components, and the final column $\mathbf{r} \in \mathbb{F}_2^n$ records the sign (phase) information via
\[
g_i \ket{\psi} = (-1)^{r_i} \ket{\psi}.
\]
The overlap between two stabilizer states $\psi_{1,2}$, with stabilizer group $\mathcal{S}_{1,2}$ respectively is \cite{kueng2015qubit}
\begin{equation}
\label{eq:overlap}
    \langle \psi_1,\psi_2\rangle\,=\,\left\{
    \begin{matrix}
        2^{-n}|\mathcal{S}_1\cap\mathcal{S}_2|& \text{ if all phases match on } \mathcal{S}_1\cap \mathcal{S}_2\\
        0 & \text{otherwise}
    \end{matrix}
    \right.
\end{equation}
Clifford gates preserve the Pauli group under conjugation. That is, if $U$ is a Clifford gate and $P \in \mathcal{P}_n$, then $U P U^\dagger \in \mathcal{P}_n$. Hence, applying a Clifford gate to a stabilizer state corresponds to updating its stabilizer generators by conjugation
\[
g_i \mapsto U g_i U^\dagger.
\]

This action can be represented by updating the tableau, and each Clifford gate has an efficient tableau update rule. For example:
\begin{itemize}
    \item \textbf{Hadamard gate} $H_j$ swaps the $X$ and $Z$ components for qubit $j$ in each generator and flips the sign if both are 1.
    \item \textbf{Phase gate} $S_j$ maps $X_j \mapsto Y_j$, i.e., adds the $X$ column to the $Z$ column for qubit $j$.
    \item \textbf{CNOT gate} $\text{CNOT}_{j,k}$ maps
    \[
    X_j \mapsto X_j X_k, \quad Z_k \mapsto Z_j Z_k.
    \]
    This corresponds to row-wise XOR operations on the appropriate $X$ and $Z$ bits.
\end{itemize}

Each of these gate operations can be implemented by a series of bitwise row and column manipulations on the tableau. Importantly, the update cost is $O(n^2)$ per gate, which makes simulation of Clifford circuits scalable. In addition to representing states, tableaux can also encode Clifford gates themselves by considering their action on a set of Pauli basis elements. A Clifford gate $U$ acts on the Pauli group via conjugation:
\[
U P_i U^\dagger = P'_i, \quad P_i \in \mathcal{P}_n,
\]
and this mapping can be described by a $2n \times 2n$ symplectic matrix over $\mathbb{F}_2$. The gate tableau tracks how each $X_j$ and $Z_j$ basis element transforms, allowing the gate to be applied to any stabilizer state tableau via matrix multiplication and sign rule updates.\\

\section{\texorpdfstring{$\mathbf{2n}$}{2n}-universal circuit}
\label{app:circ}
Here we report the $2n$ $t$-doped reconstructing circuit.
Despite the apparent "parallel" structure of the circuit note that it is completely equivalent to the series $T$-doping scheme showed in main text: one can always move the $T$ gates to have them only act on a single qubit in succession, by adding suitable SWAP operators.

\begin{figure}[tbh]
    \centering
    \includegraphics[width=1\linewidth]{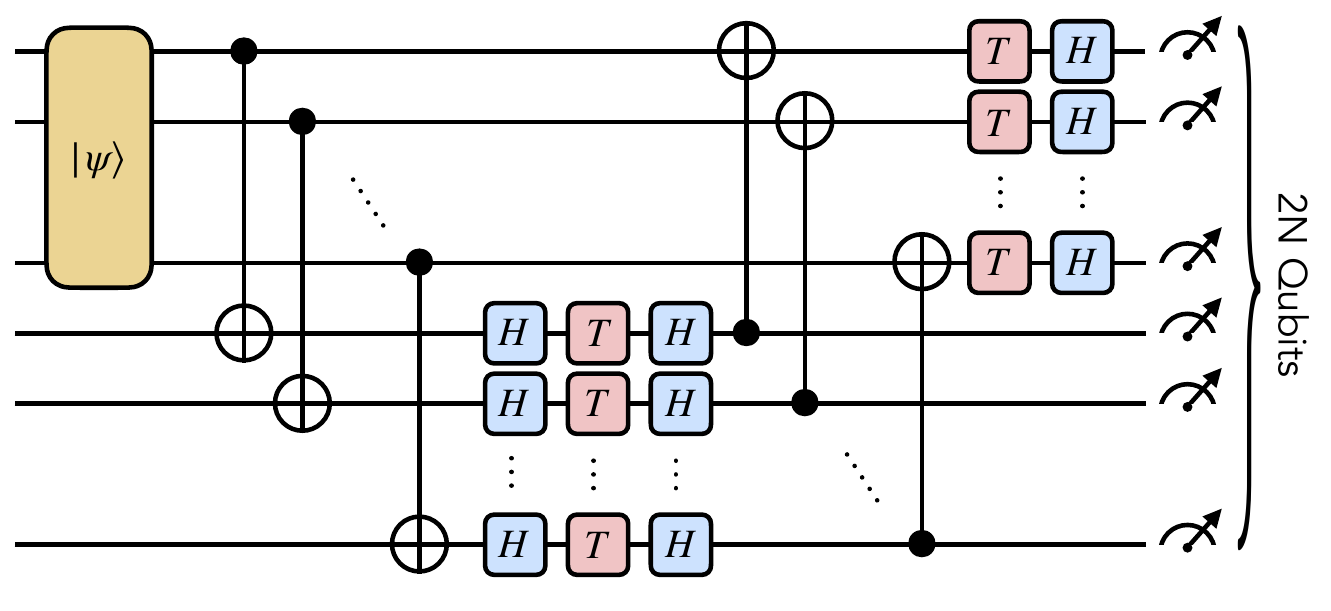}
    \caption{Example of $t=2n$-doped circuit over $n$ qubits, with $m=n$ ancillas, which gives an IC measurement.}
    \label{fig:2n-circ}
\end{figure}
The unitary operator corresponding to the circuit in \cref{fig:2n-circ} can be written as $U=\otimes_{n=1}^N U_n$ with
\begin{equation}
    U_n=(\text{H}\text{T})_n\text{CX}(n{+}N,n)(\text{HTH})_{n{+}N}\text{CX}(n,n{+}N)
\end{equation}
The elements of the POVM $\bs\mu$ implemented by the circuit is given by all the possible tensor product between operators 
\begin{equation}
\mu_{kj}=\text{Tr}\big\{(\mathbb{I}_n\otimes \bra{0}_{n+N})U_n^\dagger\mathbb{P}_{kj}^nU_n(\mathbb{I}_n\otimes \ket{0}_{n+N})\big\}
\end{equation}
with $k,j=\{0,1\}$. 
So it is sufficient to look at the invertibility of the frame operator for the case $n=1$, in this case given by
\begin{equation}
    F=\sum_{k,j=0}^{1}\ket*{\mu_{kj}}\!\!\bra*{\mu_{kj}}
    =\begin{pmatrix}
    1/2 & 0 & 0 & 0\\
    0 & 1/8 & 0 & 0\\
    0 & 0 & 1/8 & 0\\
    0 & 0 & 0 & 1/4\\
    \end{pmatrix}
\end{equation}

\section{Z-free centralizers}
\label{sec:Zfree_centralisers_count}

We now prove the general formula to count the number of $\ps{Z}$-free elements in the centraliser of a given Abelian group.

\begin{theorem}
Let $H\equiv \langle H_1,..., H_\ell\rangle \le \tildecalP_t$ an abelian subgroup of $t$-qubit Pauli strings with generators $H_i$, and let $C(H)\le\tildecalP_t$ be its centraliser.
Let $\tildecalQ\equiv\{\ps{I},\ps{X},\ps{Y}\}$, and let $\tildecalQ_t = \tildecalQ^{\times t}$ be the subset of $\ps{Z}$-free $t$-qubit strings.
Then the number of $\ps{Z}$-free elements in $C(H)$ is
\begin{equation}
    \lvert \tildecalQ_t\cap C(H)\rvert =
    \frac{1}{\lvert H\rvert}\sum_{h\in H}
    3^{n_I(h)}(-1)^{n_Z(h)}.
\end{equation}
\end{theorem}
\begin{proof}
An equivalent way to write the number of such elements is
\begin{equation}
    \lvert \tildecalQ_t\cap C(H) \rvert
    = \sum_{P\in\tildecalS_t} \prod_{i=1}^\ell\delta_{P\in C(H_i)},
\end{equation}
where $\delta_{P\in C(H_i)}=1$ iff $[P,H_i]=0$.
Note that any $s\in\{0,1\}$ can be equivalently written as $s=\frac{1+(-1)^{s+1}}{2}$, and furthermore that $\delta_{P\in C(H_i)}=1-\langle P,H_i\rangle$, where we defined the \textit{symplectic inner product} such that $\langle P,Q\rangle=0$ iff $[P,Q]=0$ and $\langle P,Q\rangle=1$ iff $\{P,Q\}=0$. This also satisfies $\langle \prod_i H_i,P\rangle=\sum_i \langle H_i,P\rangle$, and $\langle H,P\rangle=\sum_{k=1}^t \langle H_k, P_k\rangle$ with $H_k,P_k$ the single-qubit operators. Thus
\begin{equation}
    \delta_{P\in C(H)}=\prod_{i=1}^\ell\delta_{P\in C(H_i)}
    = \frac1{2^\ell}\sum_{h\in H}(-1)^{\langle h, P\rangle},
\end{equation}
and summing over $P\in\tildecalS_t$,
\begin{equation}\label{eq:Zfree_count}
\begin{aligned}
    \lvert\tildecalS_t\cap C(H)\rvert
    &= \frac{1}{2^{\ell}}\sum_{h\in H}
    \sum_{P\in\tildecalS_t}\prod_{k=1}^t
    (-1)^{\langle h_k,P_k\rangle}
    \\
    &= \frac1{2^\ell}\sum_{h\in H}
    \prod_{k=1}^t \sum_{P\in\{I,X,Y\}}
    (-1)^{\langle h_k,P \rangle}
    \\
    &= \frac{1}{\lvert H\rvert}\sum_{h\in H}
    3^{n_I(h)}(-1)^{n_Z(h)}.
\end{aligned}
\end{equation}
Thus the number of $\ps{Z}$-free strings in the centraliser of an abelian group $H$ equals the average of $3^{n_I(h)} (-1)^{n_Z(h)}$ over the elements of $H$.
\end{proof}


\parTitle{Single generator}
When there is a single generator, $H=\langle H_1\rangle=\{I,H_1\}$,~\cref{eq:Zfree_count} reduces to
\begin{equation}
    \lvert\tildecalS_t\cap C(H)\rvert =
    \frac12(3^t+3^{n_I(H_1)}(-1)^{n_Z(H_1)}).
\end{equation}
For example, the centralizer of $H=IIZX$ contains $\frac{3^4 - 3^2}{2}=36$ $\ps{Z}$-free elements, $H=IZZX$ has $\frac{3^4+3^2}{2}=45$, and $H=IIII$ has $\frac{3^4+3^4}{2}=3^4$.

\parTitle{Two generators}
If $H=\langle H_1,H_2\rangle$, we get
\begin{equation}
\begin{gathered}
    \frac14\big[3^t + 3^{n_I(H_1)}(-1)^{n_Z(H_1)}
    + 3^{n_I(H_2)}(-1)^{n_Z(H_2)}
    \\+ 3^{n_I(H_1 H_2)}(-1)^{n_Z(H_1 H_2)}
    \big],
\end{gathered}
\end{equation}
where we notice that $n_I(H_1 H_2)$ is also equal to the number of positions where $H_1$ and $H_2$ have the same operator, while $n_Z(H_1 H_2)$ is the number of posititions where the two generators have one of the pairs $(I,Z),(Z,I),(X,Y),(Y,X)$.
For example if $H_1=IIXZ$ and $H_2=XZII$, then we get $\frac{3^4-2\times 3^2 + 1}{4}=16$.

\section{Examples}
\label{sec:examples}

We give here some explicit example applications of~\cref{thm:structural_clifford_povms_plus}.
\begin{example}[label=ex:thm2_1]{}{}
Let $\calS\equiv\langle \ps{ZZI},\ps{ZIZ},\ps{XXX}\rangle$ with $n=2$, $m=1$, and $\ket\psi=\ket+$, corresponding to $\calZ=\langle \ps{X}\rangle$, $\calZ'=\langle \ps{IIX}\rangle$.

In this case $\calS\cap \calZ'=\{\ps{I}\}$, because $\ps{IIX}\notin\calS$, but $\calS\cap C(\calZ')=\langle \ps{ZZI}, \ps{XXX}\rangle$. In the notation of~\cref{thm:structural_clifford_povms_plus} we have $\ell=0$ because of the trivial intersection, the generator of $\calS$ anticommuting with $\calZ'$ is $\tilde g_1= \ps{ZIZ}$, a pair of generators of $\calS$ commuting with $\calZ'$ but not in it is $g_1=\ps{ZZI}$, $g_2=\ps{XXX}$, and finally $\tilde h_1=\ps{X}$ is in $\calZ'$ but not in $\calS$.

The observables that are reconstructed by the measurement are the linear span of the projections of the generators $g_k$, namely $\pi_n(g_1)=\ps{ZZ}$ and $\pi_n(g_2)=\ps{XX}$.
In other words, the effective POVM is equivalent to measuring the stabilizer basis specified by the group $\langle \ps{ZZ},\ps{XX}\rangle$. More explicitly, it is an eight-outcome POVM with the four distinct elements $\frac12(\frac{I\pm ZZ}{2})(\frac{I\pm XX}{2})$, each repeated twice.

Using the same $\calS$ but now with $n=1, m=2$, and $\calZ=\langle \ps{XX},\ps{YY}\rangle$, we see that $\calS\cap\calZ'=\langle \ps{IZZ}\rangle$, $\calS\cap C(\calZ')=\langle \ps{XXX}\rangle$, $\tilde g_1=\ps{ZZI}$. Thus the effective POVM has elements $\frac12(\frac{I\pm X}{2})$, each one repeated twice, and four other vanishing elements.
\end{example}

\begin{example}[label=ex:thm2_2]{}{}
Let $\calS=\langle \ps{ZXZZY}, \ps{XIIIZ}, \ps{XIIZZ}, \ps{ZXYZX}, \ps{ZIYZX}\rangle$, with $n=2$, $m=3$ and $\calZ=\langle \ps{ZZI},\ps{ZIZ},\ps{XXX}\rangle$. Then $\calS\cap C(\calZ')=\langle \ps{IXIII},\ps{XIIZZ} \rangle$, which has trivial intersection with $\calZ'$. Thus the effective measurement is informationally equivalent to measuring $\langle \ps{XI},\ps{IX}\rangle$ on $\calH_n$, and has elements $\frac18(\frac{I\pm XI}{2})(\frac{I\pm IX}{2})$ each repeated eight times.
\end{example}

We now consider some examples to illustrate the relation between double coset decompositions and supported $\dimspanmu$ discussed in~\cref{thm:entanglement_and_reconstruction,cor:rank_multiples}.
\begin{example}[label=ex:ent1]{Double coset decompositions}{}
    Let $\calS=\langle \ps{XIII},\ps{IIXI},\ps{IXIX},\ps{IYIY}\rangle$, $n=m=2$.
    Then we have the decomposition~\cref{eq:entanglement_relations} with $\calS_n=\langle\ps{XIII}\rangle$, $\calS_m=\langle\ps{IIXI}\rangle$, $a_1=b_1=\ps{XI}$, $g_1^{(n)}=g_1^{(m)}=\ps{IX}$, $\bar g_1^{(n)}=\bar g_1^{(m)}=\ps{IY}$. Thus $p=1$, and the maximal supported $\dimspanmu$ is $2^3$, compatibly with the three independent operators on $\calH_n$: $\ps{XI}$, $\ps{IX}$, and $\ps{IY}$.
    
    Explicitly, the double-coset decomposition discussed in~\cref{cor:rank_multiples} in this case results in $4$ cosets, each one containing $2$ distinct operators on the $\calH_n$ side:
    \begin{equation}
    \begin{aligned}
        \calS/\calS_n\calS_m=\{
            &\{\ps{II},\ps{XI}\}\times\{\ps{II},\ps{XI}\}, \\
            &\{\ps{IX},\ps{XX}\}\times\{\ps{IX},\ps{XX}\}, \\
            &\{\ps{IY},\ps{XY}\}\times\{\ps{IY},\ps{XY}\}, \\
            &\{\ps{IZ},\ps{XZ}\}\times\{\ps{IZ},\ps{XZ}\}
        \}.
    \end{aligned}
    \end{equation}
    For a given $\ket\psi$ to achieve $\dimspanmu=2^3$ it must have nonzero expectation value on at least one element of each of the three non-identity $\calH_m$ cosets, namely: $\{\ps{IX},\ps{XX}\}$, $\{\ps{IY},\ps{XY}\}$, and $\{\ps{IZ},\ps{XZ}\}$.
    One such example is $\mathbb{P}_\psi=\mathbb{P}_0\otimes \frac{I+(X+Y+Z)/\sqrt3}{2}$ which by construction has nonzero expectation values on all local Paulis on its second qubit.
    By contrast, using $\ket\psi=\ket T^{\otimes2}$ would instead only give $\dimspanmu=6$, because $\{\ps{IZ},\ps{XZ}\}$ does not survive.
    And using a stabiliser state as projection, for example $\ket\psi=\ket{00}$, gives $\dimspanmu=2^2$ because $\{\ps{IZ},\ps{XZ}\}$ is the only (non-identity) coset that survives the projection.
\end{example}

\begin{example}[label=ex:ent2]{Double coset decompositions}{}
Consider
\begin{equation*}
\begin{aligned}
    \calS
    &=\langle \ps{ZXZZY}, \ps{XIIIZ}, \ps{XIIZZ}, \ps{ZXYZX}, \ps{ZIYZX}\rangle \\
    &=\langle \ps{IXIII}, \ps{IIIZI}, \ps{IIXIZ},
    \ps{ZXZZY}, \ps{XIIIZ} \rangle,
\end{aligned}
\end{equation*}
with $n=2$, $m=3$.
Then $\calS_n=\langle \ps{IXIII}\rangle$, $\calS_m=\langle \ps{IIIZI}, \ps{IIXIZ}\rangle$, $a_1=\ps{IX}$, $b_1=\ps{IZI}$, $b_2=\ps{XIZ}$, $g_1^{(n)}=\ps{ZX}$, $g_1^{(m)}=\ps{ZZY}$, $\bar g_1^{(n)}=\ps{XI}$, $\bar g_1^{(m)}=\ps{IIZ}$, and thus $p=1$. The double coset decomposition reads
\begin{equation*}
\begin{aligned}
    \calS/\calS_n\calS_m = \{
    &\{\ps{II},\ps{IX}\} \times \{\ps{III},\ps{IZI},\ps{XIZ},\ps{XZZ}\}, \\
    &\{\ps{ZX},\ps{ZI}\}\times
    \{\ps{ZZY},\ps{ZIY},\ps{YZX},\ps{YIX}\}, \\
    &\{\ps{XI},\ps{XX}\}\times
    \{\ps{IIZ},\ps{IZZ},\ps{XII},\ps{XZI}\}, \\
    &\{\ps{YX},\ps{YI}\}\times
    \{\ps{ZZX},\ps{ZIX},\ps{YZY},\ps{YIY}\}
    \}.
\end{aligned}
\end{equation*}
The maximal supported $\dimspanmu$ is $2^{n+p}=2^3$, achieved \textit{e.g.} with $\ket\psi=\ket T^{\otimes3}$, and each of these four cosets contains at least one $\ps{Z}$-free element in its $\calH_m$ component.
\end{example}

\begin{example}[label=ex:ent3]{Double coset decompositions}{}
    Let $\calS=\langle \ps{IXZ},\ps{XYY},\ps{YZY}\rangle$, $n=1$, $m=2$.
    Then $\calS_n=\{\ps{I}\}$, $\calS_m=\langle \ps{IXZ}\rangle$, hence $a_1=\ps{IX}$, $b_1=\ps{IXI}$, $b_2=\ps{XIZ}$, $g_1^{(n)}=\ps{X}$, $g_1^{(m)}=\ps{YY}$, $\bar g_1^{(n)}=\ps{Y}$, $\bar g_1^{(m)}=\ps{ZY}$. In this case $p=n=1$ thus we have maximal entanglement.
    The double coset decomposition reads
    \begin{equation*}
    \begin{aligned}
        \calS/\calS_n\calS_m = \{
        &\{\ps{I}\} \times \{\ps{II},\ps{XZ}\},
        \{\ps{X}\} \times \{\ps{YY},\ps{ZX}\} \\
        &\{\ps{Y}\} \times \{\ps{ZY},\ps{YX}\},
        \{\ps{Z}\} \times \{\ps{XI},\ps{IZ}\}
        \}.
    \end{aligned}
    \end{equation*}
    and thus the maximal supported $\dimspanmu$ is $2^{n+p}=4$, corresponding to an IC measurement, and is achievable with $\ket\psi=\ket T^{\otimes 2}$.
\end{example}

Consider now an explicit example of how the coset $\calS_m$ is sufficient to characterize $\dimspanmu$ for any given $\ket\psi$, as discussed in~\cref{sec:entanglement_and_reconstruction}:

\begin{example}[label=ex:ent4]{}{}
    Let $n=1, m=2$, $\calS_m=\langle \ps{IXZ}\rangle$. Considering the quotient of the centraliser of the projection $\pi_m(\calS_m)=\langle\ps{XZ}\rangle$ over itself, we get
    \begin{equation*}
    \begin{aligned}
        C(\pi_m(\calS_m))/\pi_m(\calS_m) = 
        \{&\{\ps{II},\ps{XZ}\}, \{\ps{XI},\ps{IZ}\},
        \\
        &\{\ps{YY},\ps{ZX}\}, \{\ps{YX},\ps{ZY}\}
        \}.
    \end{aligned}
    \end{equation*}
    Then, as previously discussed, regardless of how each of these cosets is attached to a Pauli operator on $\calH_n$, we immediately know that there are four independent cosets and thus projecting on random states we get an IC-POVM.
    And similarly we know that projecting onto $\ket T^{\otimes2}$ still gives an IC-POVM, because all four cosets contain $\ps{Z}$-free elements.
\end{example}
\begin{example}[label=ex:ent5]{}{}
    Let $n=1, m=2$, $\calS_m=\langle \ps{IXX}\rangle$. Then
    \begin{equation*}
    \begin{aligned}
        C(\pi_m(\calS_m))/\pi_m(\calS_m) = 
        \{&\{\ps{II},\ps{XX}\}, \{\ps{XI},\ps{IX}\},
        \\
        &\{\ps{YY},\ps{ZZ}\}, \{\ps{YZ},\ps{ZY}\}
        \}.
    \end{aligned}
    \end{equation*}
    Thus with $\ket\psi=\ket T^{\otimes2}$ we get $\dimspanmu=3$, because $\{\ps{YZ},\ps{ZY}\}$ gets annihilated by the projection.
\end{example}

We include here a full proof of the $n=m=2$ case of~\cref{conj:maxentmaxrank}:
\begin{example}[label=ex:ent6]{}{}
    We will prove the conjecture in the special case of $n=m=2$.
    The approach used here is, however, extremely \textit{ad-hoc} --- and somewhat unsatisfying. It is likely that some more general property of pure states would need to be used to handle the general case.

    Consider the special case with $n=m=2$, $p=1$. Then
    $\calS=\langle
    \ps{ZIII}, \ps{IIZI},
    \ps{IXIX}, \ps{IYIY}
    \rangle$, and the effective measurement operators read
    \begin{equation*}
        (Z^a\otimes P)\,
        \langle\psi|
        Z^b\otimes P
        |\psi\rangle,
    \end{equation*}
    with $a,b\in\{0,1\}$, $P\in\{I,X,Y,Z\}$.
    If instead $p=2$, the measurement operators are $(P\otimes Q)\, \langle\psi|P\otimes Q|\psi\rangle$ for all $16$ combinations of $P,Q\in\{I,X,Y,Z\}$.
    
    If there are $r=2$ nonzero expectation values, then $\dimspanmu=2r=4$, which we know from~\cref{lemma1} is the smallest possible $\dimspanmu$, so certainly increasing to $p=2$ does not decrease $\dimspanmu$.
    If instead $r=3$, then $\dimspanmu=6$, but then $\ket\psi$ is not a stabilizer state, and thus $\dimspanmu'\ge 6$.
    Finally, if $r=4$, then
    $|\langle IX\rangle|+|\langle ZX\rangle|$,
    $|\langle IY\rangle|+|\langle ZY\rangle|$,
    $|\langle IZ\rangle|+|\langle ZZ\rangle|$, are nonzero.
    Any pure two-qubit state can be written as $\rho=\frac{I+A}{4}$ where
    \begin{equation*}
    \begin{gathered}
        A =
        \sum_{i=1}^3 a_i(\sigma_i\otimes I)+\sum_{j=1}^3 b_j(I\otimes\sigma_j)
        +
        \sum_{i,j=1}^3 T_{ij}\sigma_i\otimes\sigma_j,
    \end{gathered}
    \end{equation*}
    with the coefficients $(a_i)_i$ characterising the reduced state $\rho_A$, $(b_j)_j$ characterising $\rho_B$, and $T_{ij}$ characterising the correlations.
    For pure states these coefficients are such that $Tb=a$, $b=T^T a$, $TT^T=(1-\|a\|^2)I+aa^T$, $T^T T=(1-\|b\|^2)I+bb^T$.
    Thus
    \begin{itemize}
    \item If the state is separable then the total number of nonzero expectation values is at least $8$, because the assumption forces $\rho_B$ to have nonzero expectation values on all four Paulis $\{I,X,Y,Z\}$, and purity forces $\rho_A\neq I/2$.
    \item If $b_i\neq 0$ for any $i$, then the $i$-th column of $T$ must also be nonzero.
    Therefore $3$ nonzero $b_i$ imply $3$ nonzero $T_{ij}$, and at least one nonzero $a_i$, hence $r'\ge 8$.
    \item If $b=0$, by assumption $T_{3j}\neq0$ for all $j$, $a=0$, and $T\in\mathbf{SO}(3)$. Thus if $T$ has a fully nonzero row, then all its elements must be nonzero, hence $r'\ge10>8$.
    \item If $\rho$ is entangled then the constraints on $TT^T$ and $T^T T$ force $T$ to have pairwise orthogonal and nonzero rows (and columns). This is sufficient to handle the rest of the cases.
    For example if the only nonzero local terms are $b_1,a_1$, then $T_{32},T_{33}\neq 0$, and then at least other 3 nonzero elements of $T$ must be nonzero for its rows to be nonzero and orthogonal.
    The other cases are handled similarly.
    \end{itemize}
\end{example}

We now provide some application examples of the gadget formalism discussed in~\hyperlink{ref:tdoped_part1}{section V}, and in particular of how Heisenberg-evolved measurement states for $t$-doped circuits can be characterized via stabilizer groups with some fixed syndromes.

\begin{example}[label=ex:tdoped_1]{}{}
Let $n=m=1$ and $t=2$. Suppose the measurement, before projecting onto $\ket\psi$, is $\calS=\langle \ps{ZIII},\ps{IZII},\ps{IIXX},\ps{IIYY} \rangle$, with fixed syndromes $\ps{IIXX}=\ps{IZII} =1$.
Projecting the ancilla onto any stabilizer state fixes $\ps{IZII}$.
This is therefore a best-case scenario in the terminology above, meaning that the ancilla projection did not affect any of the two free generators, which here are $\ps{ZIII}$ and $\ps{IIYY}$. The remaining measurement on the $n+t=3$ qubits is thus described by $\langle \ps{ZII},\ps{IXX},\ps{IYY}\rangle$ with $\ps{IXX}=1$.
\end{example}

\begin{example}[label=ex:tdoped_2]{Explicit circuit producing $\calS$}{}
Consider $n=m=1$, $t=2$, the $t$-doped circuit
$U=H_1 T_1 \mathrm{CX}_{2\to1} H_2 T_2 H_2 \mathrm{CX}_{1\to2}$, computational basis measurements at the output, and initial ancilla state $\ket\psi=\ket0$.
This is the circuit pictured in~\cref{fig:explicit_gadget_example}.

The physical measurement in the gadget picture is described by $\langle \ps{ZIII},\ps{IZII},\ps{IIZI},\ps{IIIZ}\rangle$, with fixed syndromes $\ps{IIZI}=\ps{IIIZ}=1$. Evolving these operators through the circuit gives
\begin{equation*}
\begin{aligned}
    &\tt \ps{ZIII} \to \ps{XXIX},
    \qquad \ps{IZII} \to \ps{ZZXI}, \\
    &\tt \ps{IIZI} \to \ps{IXZI},
    \qquad \ps{IIIZ} \to \ps{IZXZ},
\end{aligned}
\end{equation*}
and thus $\calS=\langle  \ps{XXIX}, \ps{ZZXI},\ps{IXZI},\ps{IZXZ}\rangle$ with $\ps{IXZI}=\ps{IZXZ}=1$ describes the measurement before projection.
Using the result of~\cref{thm:structural_clifford_povms_plus}, given $\calZ'=\langle \ps{IZII}\rangle$, we observe that $\calS\cap\calZ'=\{\ps{I}\}$ and $\calS\cap C(\calZ')=\langle \ps{XIZX}, \ps{ZZXI},\ps{IZXZ} \rangle$. Thus projecting the ancilla on $\ket0$ (or $\ket1$) gives a measurement on the $n+t$ qubits characterised by $\langle \ps{XZX}, \ps{ZXI}, \ps{IXZ}\rangle$ with $\ps{IXZ}=1$.
In particular, these are maximally entangled states with $p=1$, and induce the coset decomposition
\begin{equation*}
\begin{aligned}
    \calS/\calS_t=\{
    &\tt \{\ps{III},\ps{IXZ}\},\{\ps{ZXI},\ps{ZIZ}\}, \\
    &\tt \{\ps{XYY},\ps{XZX}\}, \{\ps{YYX},\ps{YZY}\}\}.
\end{aligned}
\end{equation*}
Note in particular that there is a single $\ps{Z}$-free operator in each coset, which means that projecing on $\ket{T}^{\otimes 2}$ all cosets survive, hence the resulting measurement is IC.
\end{example}

\begin{example}[label=ex:tdoped_3]{}{}
    Going back to~\cref{ex:ent4}, where $n=2$, $t=1$, $\calS=\langle \ps{ZZI},\ps{ZIZ},\ps{XXX}\rangle$,~\cref{thm:rank_tlen} predicts $\dimspanmu=2^2(\frac32)=6$. Thus using a single $T$ gate gives a value of $\dimspanmu$ between the $2^2=4$ obtained for stabilizer ancillas, and the $2^4=16$ achievable with optimal initial ancillas.
    To see the theorem in action more explicitly, we have $\calS_n=\langle\ps{ZZI}\rangle$, hence
    \begin{equation*}
    \begin{aligned}
        \calS/\calS_n = \{
        &\{\ps{III}, \ps{ZZI}\},
        \{\ps{ZIZ}, \ps{IZZ}\}, \\
        &\{\ps{XXX}, \ps{YYX}\},
        \{\ps{YXY}, \ps{XYY}\}\}.
    \end{aligned}
    \end{equation*}
    Looking at the last qubit in each coset, it becomes evident that $3$ out of the $4$ cosets survive the projection onto $\ket T$, which entirely annihilates the coset $\{\ps{ZIZ},\ps{IZZ}\}$. Hence $\dimspanmu=6$.
    This then generalises because under the assumption of maximal entanglement and $t\le n$, it is always true that each coset corresponds to a unique element $\pi_t(g)$, and that all strings in $\calH_t$ appear in some coset.
\end{example}

\end{document}